\documentclass[11pt,a4paper,titlepage,oldfontcommands]{memoir}

\usepackage[OT1]{fontenc}

\usepackage[english]{babel}

\usepackage[utf8]{inputenc}

\usepackage[sc]{mathpazo}

\usepackage{amsmath,amssymb,amsfonts,mathrsfs}

\usepackage[amsmath,thmmarks]{ntheorem}

\usepackage{graphicx}

\usepackage{soul}

\usepackage{pdfpages}

\usepackage{varioref}

\usepackage{datetime}

\usepackage{mathtools}

\usepackage[h]{esvect}

\usepackage{array}

\usepackage{listings}
\usepackage[activate]{pdfcprot}

\usepackage{booktabs}

\usepackage{enumitem}
\usepackage{verbatim}
\usepackage{pgfplots}

\usepackage{ETHlogo}

\nonzeroparskip
\defaultlists

\makeatletter

\if@twoside
  \copypagestyle{chapter}{Ruled}
\else
  \copypagestyle{chapter}{ruled}
\fi
\makeoddhead{chapter}{}{}{}
\makeevenhead{chapter}{}{}{}
\makeheadrule{chapter}{\textwidth}{0pt}
\copypagestyle{abstract}{empty}

\makechapterstyle{bianchimod}{%
  \chapterstyle{default}
  \renewcommand*{\chapnamefont}{\normalfont\Large\sffamily}
  
  \renewcommand*{\printchaptername}{%
    \chapnamefont\centering\@chapapp}

  }

\chapterstyle{bianchimod}

\setsecheadstyle{\Large\bfseries\sffamily}
\setsubsecheadstyle{\large\bfseries\sffamily}
\setsubsubsecheadstyle{\bfseries\sffamily}
\setparaheadstyle{\normalsize\bfseries\sffamily}
\setsubparaheadstyle{\normalsize\itshape\sffamily}
\setsubparaindent{0pt}

\captionnamefont{\sffamily\bfseries\footnotesize}
\captiontitlefont{\sffamily\footnotesize}
\setsecnumdepth{subsection}
\settocdepth{subsection}

\pretitle{\vspace{0pt plus 0.7fill}\begin{center}\HUGE\sffamily\bfseries}
\posttitle{\end{center}\par}
\preauthor{\par\begin{center}\let\and\\\Large\sffamily}
\postauthor{\end{center}}
\predate{\par\begin{center}\Large\sffamily}
\postdate{\end{center}}

\def\@advisors{}
\newcommand{\advisors}[1]{\def\@advisors{#1}}
\def\@department{}
\newcommand{\department}[1]{\def\@department{#1}}
\def\@thesistype{}
\newcommand{\thesistype}[1]{\def\@thesistype{#1}}

\renewcommand{\maketitlehookb}{\vspace{1in}%
  \par\begin{center}\Large\sffamily\@thesistype\end{center}}

\renewcommand{\maketitlehookd}{%
  \vfill\par
  \begin{flushright}
    \sffamily
    \@advisors\par
    \@department, ETH Z\"urich
  \end{flushright}
}

\checkandfixthelayout

\makeatother

\theoremstyle{plain}
\numberwithin{equation}{chapter}

\newtheorem{theorem}{Theorem}[chapter]
\newtheorem{example}[theorem]{Example}
\newtheorem{remark}[theorem]{Remark}

\newtheorem{definition}[theorem]{Definition}

\newtheorem{proposition}[theorem]{Proposition}
\theorembodyfont{\normalfont}
\newtheorem{algorithm}[theorem]{Algorithm}

\theoremstyle{nonumberplain}
\theorembodyfont{\normalfont}
\theoremsymbol{\ensuremath{\square}}
\newtheorem{proof}{Proof}

\renewcommand{\epsilon}{\ensuremath\varepsilon}

\renewcommand{\phi}{\ensuremath{\varphi}}

\usepackage[linkcolor=black,colorlinks=true,citecolor=black,filecolor=black]{hyperref}

\title{\huge Copula-based hierarchical risk aggregation \\ \vspace{0.5cm} \Large \textnormal{Tree dependent sampling and the space of mild tree dependence}}
\author{Fabio Derendinger}
\thesistype{Master Thesis}
\advisors{Advisors: Prof.\ Dr.\ Paul Embrechts,\\ Prof.\ Dr.\ Hans-Jürgen Wolter,\\ Dr.\ Philipp Arbenz}
\department{Department of Mathematics}
\date{Monday 8\textsuperscript{th} March, 2015}

\begin{document}

\frontmatter

\begin{titlingpage}
  \calccentering{\unitlength}
  \begin{adjustwidth*}{\unitlength-24pt}{-\unitlength-24pt}
    \maketitle
  \end{adjustwidth*}
\end{titlingpage}

\begin{abstract}
The ability to adequately model risks is crucial for insurance companies. The method of "Copula-based hierarchical risk aggregation" (Arbenz et al. \cite{arbenz12}) offers a flexible way in doing so and has attracted much attention recently. We briefly introduce the aggregation tree model as well as the sampling algorithm proposed by they authors.\\ \\
An important characteristic of the model is that the joint distribution of all risk is not fully specified unless an additional assumption (known as "conditional independence assumption") is added. We show that there is numerical evidence that the sampling algorithm yields an approximation of the distribution uniquely specified by the conditional independence assumption. We propose a modified algorithm and provide a proof that under certain conditions the said distribution is indeed approximated by our algorithm. \\
We further determine the space of feasible distributions for a given aggregation tree model in case we drop the conditional independence assumption. We study the impact of the input parameters and the tree structure, which allows conclusions of the way the aggregation tree should be designed.
\end{abstract}
\newpage
\renewcommand{\abstractname}{Acknowledgements}
\begin{abstract}
I would like to thank my supervisor, \textit{Prof. Dr. Hans-Jürgen Wolter}, as well as my co-supervisor, \textit{Prof. Dr. Paul Embrechts}, for introducing me to the subject and for the encouragement. Also, I wish to thanks \textit{Hansjörg Furrer} and \textit{Christoph Möhr} (FINMA) for their helpful comments in the early stage of this thesis. \\
Finally, I would particularly like to thank my co-supervisor, \textit{Dr. Philipp Arbenz (SCOR)}, who has supported me throughout my thesis with his patience and knowledge. Without his expertise and his continuous encouragement this thesis would not have been possible.
\end{abstract}

\cleartorecto
\tableofcontents
\mainmatter

\newcommand{\package}{\emph}

\chapter{Introduction}
\label{cha:1}

Why modelling risk? From the perspective of an insurer or reinsurer the answer to this question is straightforward: Taking risks is their core business, and hence both their profitability and solvency critically depends on the ability to adequately model these risks. Besides that, a proper risk management is also required by regulatory frameworks such as Solvency II and Basel III.
\\ \\
Mathematically, risks can be interpreted as a multivariate random variable $\boldsymbol{X}=(X_1,\ldots,X_n)$, where the univariate random variables $X_1, \ldots, X_n$ represent the individual risks (or marginal risks). Given that the individual risks often share common environmental and socioeconomic conditions, they are generally dependent. Therefore, most of the time knowledge of the joint distribution is required to be able to properly measure and allocate risk.
\\ \\
Assume, for instance, an insurance company is interested in the total amount of claim payments in a given period in the future, i.e. the quantity $X_{\varnothing}:=X_1+\ldots+X_n$. In order to compute this quantity, it does not suffice to merely know the distributions of the individual risks. Presumably the most obvious way to determine the distribution of a sum of dependent risk is to first determine the joint distribution function $F(x_1,\ldots,x_n)=P[X_1\leq x_1,\ldots, X_n \leq x_n]$ of the individual risks. Modelling this distribution accurately is a very challenging task. Although the individual risks constituting the portfolio might easily be described with an appropriate stochastic model derived from data and/or expert opinion, it is often the case that very few joint observations are available, in which case the joint distribution of the individual risks is basically unknown \cite{bruneton11}.
\\ \\
Recently, copulas have become the privileged tool to overcome this difficulty. Readers who are not familiar with the theory of copulas will find a brief introduction to it in Section \ref{sec:copulas}. For now, it is enough to think of a copula as a multivariate random variable that describes the dependence structure between the individual risks.\\
A well-known result in copula theory (see Theorem \ref{the:sklar}) states that the distribution function $F$ can then be written as 
\begin{align*}
F(x_1, \ldots , x_d ) = C(F_1(x_1), \ldots , F_d(x_d )),
\end{align*}
where $F_i(x)=P[X_i\leq x]$, $i=1,\ldots,n$, are the marginal distribution functions and $C:[0,1]^n\rightarrow [0,1]$ is a copula function. In this way, we have separated the dependence structure from the margins, and the above problem is reduced to find accurate models for the margins $F_1, \ldots, F_n$ and the dependence structure described through the copula $C$. When it comes to choosing the right copula, we can rely on the broad set of different copula models that have been developed and studied in the past years. In particular, there exist asymmetric copulas and copulas with tail-dependence which try to reflect the effects that can be observed in practice. \\ \\
Nevertheless, the common parametric copula models are often problematic when used in high dimensions because the attainable dependence structures are limited. For instance, they often are too symmetric.\\
A very elegant method to overcome the limitations arising in high dimensions is commonly referred to as "copula-based hierarchical risk aggregation". The method has been used in the industry for more than a decade. Consider the following simple example where we present its general idea: 

\begin{example}
\label{ex:hierarchical}
Assume we are given four different risks represented by the 4-dim. random variable $\boldsymbol{X}=(X_{1,1},X_{1,2},X_{2,1},X_{2,2})$. Here $X_{1,1}$ stands for "Car insurance Switzerland", and $X_{1,2}$ stands for "Car insurance Italy". Also, $X_{2,1}$ stands for "Earthquake Switzerland", and $X_{2,2}$ stands for "Earthquake Italy". In case we are interested in the aggregated risk $X_{\varnothing}:= X_{1,1}+X_{1,2}+X_{2,1}+X_{2,2}$, we could try to find a model for the joint distribution function $F$ of $\boldsymbol{X}$ by first modelling the marginal distributions $F_{1,1}$, $F_{1,2}$, $F_{2,1}$ and $F_{2,2}$ of the individual risk and impose a 4-dim. copula $C$ between the individual risk. The joint distribution would then be given by $F(x_1,x_2,x_3,x_4)=C(F_{1,1}(x_1),F_{1,2}(x_2),F_{2,1}(x_3),F_{2,2}(x_4))$, and the distribution of $X_{\varnothing}$ can directly be computed from $F$.\\
Alternatively, we could in a first step model the distributions of the partial risks, $(X_{1,1},X_{1,2})$ and $(X_{2,1},X_{2,2})$, by combining the risks $X_{1,1}$ and $X_{1,2}$ through a bivariate copula $C_{1}$, whereas $X_{2,1}$ and $X_{2,2}$ are combined through a bivariate copula $C_{2}$. Knowing the distributions of the partial risks, we can then easily compute the distribution of the partial sums $X_{1}:=X_{1,1}+X_{1,2}$ and $X_{2}:=X_{2,1}+X_{2,2}$. The partial sums, $X_{1}$ and $X_{2}$, can then again be combined through an adequate bivariate copula model $C_{\varnothing}$. Finally, this then allows us to compute the distribution $F_{\varnothing}$ of the total aggregate $X_{\varnothing}:=X_1+X_2=X_{1,1}+X_{1,2}+X_{2,1}+X_{2,2}$.
\begin{figure}[h]
\centering
\includegraphics[width=1\linewidth]{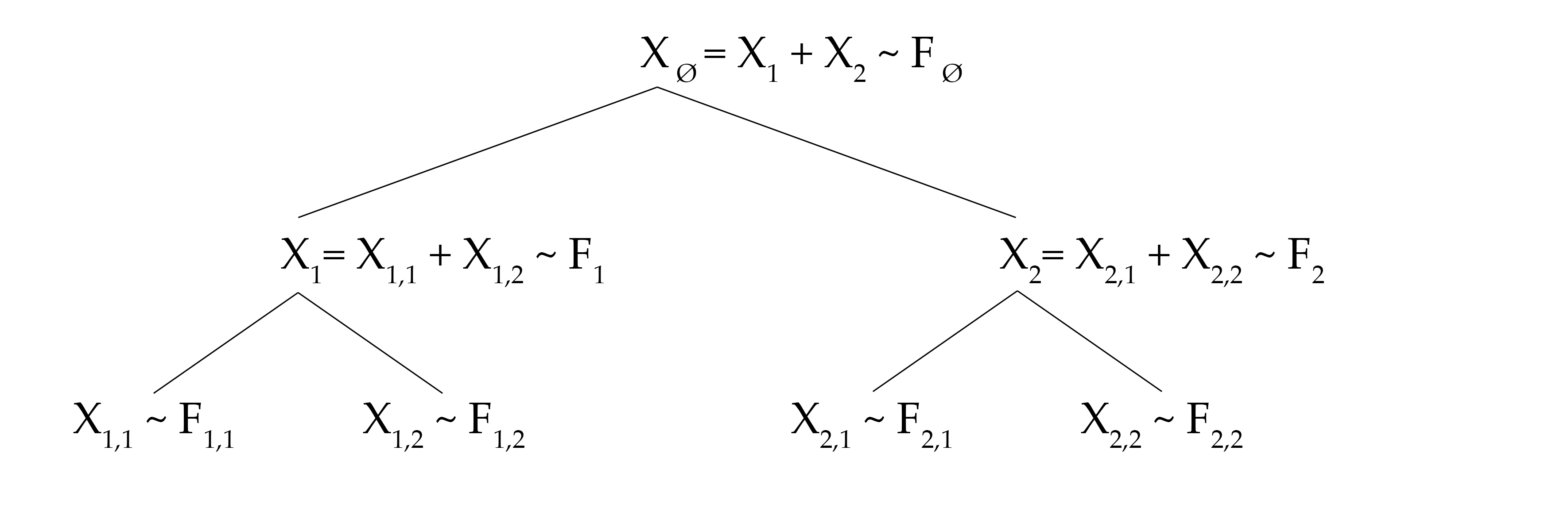}
\caption{An illustration of the 4-dim. aggregation tree model described in our lead Example \ref{ex:hierarchical}.}
\label{fig:4dimtree1}
\end{figure}

This procedure is known as "copula-based hierarchical risk aggregation" and is best illustrated by a so called "aggregation tree". The aggregation tree graphically represents in which way the individual risks are tied together and which dependence structure is assumed to hold between them. Figure \ref{fig:4dimtree1} shows the aggregations tree corresponding to the situation described above.
\end{example}
The main advantage of hierarchical risk aggregation is that we do not need to specify the copula of all risks. We once again emphasize that it is extremely unlikely to find a copula model which adequately describes the dependence structure between a large number of risks. Joint observations between all risks are too rare, and the attainable dependence structures of common parametric copula models are too limited. \\
Instead, we aggregate the risks hierarchically, which requires only to specify the joint dependence between the aggregated sub-portfolios in the different aggregation steps. For these sub-portfolios we can use low dimensional copulas, which are well-known to be more realistic measures of dependence and offer greater flexibility.  In particular, this model allows to use different dependence characteristics (tail dependence, radial asymmetry, etc.) for each aggregation step. 
\\ \\
The model can be seen as a dimension reduction tool: complexity and the number of parameters can be adjusted to the necessary intricacy and available information. As a logical consequence, the method \textit{does not} specify the full joint distribution of the individual risks. In other words, there is in general more than one distribution which satisfies a given aggregation tree model. We will call those distributions "mildly tree dependent".
\\ \\
In recent years, a number of papers have been published on this subject. Arbenz et al. \cite{arbenz12} were the first to provide a sound mathematical foundation for the copula-based hierarchical risk aggregation approach. They described the structure of the model in graph-theoretical terms and identified a condition under which it leads to a unique multivariate distribution. This condition is known as "conditional independence assumption" and the unique multivariate distribution specified by it is called "tree dependent distribution". They also provided a sampling algorithm together with a proof that the samples generated by the algorithm approximate the aggregated sub-portfolios in the model (e.g., $X_{1}$, $X_{2}$ and $X_{\varnothing}$ in our lead Example \ref{ex:hierarchical}). \\
The papers of Côté \& Genest \cite{genest15} and Bruneton \cite{bruneton11} focus on the structure of the aggregation tree. In fact a lot of freedom is simply hidden in the order in which the risks are aggregated together and in the general shape of the tree. Bruneton \cite{bruneton11} argues that the diversification effect is larger for thin trees than for fat trees. Côté \& Genest \cite{genest15} suggest a procedure for selecting the tree structure, based on hierarchical clustering techniques.  
\\ \\
The method of hierarchical risk aggregation is particularly suitable if the ultimate goal is to find an approximation for the total aggregate of all individual risks. In this case, the fact that the proposed strategy does not lead to a unique joint distribution is not critical. \\
In practice, however, we may as well be confronted with situations where the joint distribution of the individual risks is required. A popular example of such would, for instance, be risk allocation or - something more related to Example \ref{ex:hierarchical} - the situation where we want to estimate the total risk we are taking in Switzerland, i.e. the distribution of $X_{1,2}+X_{2,1}$. Unfortunately, the aggregation tree model in Figure \ref{fig:4dimtree1} does not uniquely specify this distribution.
\\ \\
The majority of the results and insights provided by the above mentioned papers cover the situation where the total aggregate $X_{\varnothing}$ is of particular interest. Concerning the joint distribution, there are a couple of interesting questions that have not yet been sufficiently addressed. This thesis aims to contribute to a better mathematical understanding of the joint distributions associated with an aggregation tree model.
\\ \\
The subject of the first question is the sampling algorithm proposed by Arbenz et al. \cite{arbenz12}. As mentioned before, the samples can be used to approximate the distribution of the total aggregate $X_{\varnothing}$ of the individual risks. However, it is uncertain so far if the samples are also an approximation of a mildly tree dependent distribution. We will conduct a numerical experiment that suggests that the samples approximate the unique tree dependent distribution specified by the conditional independence assumption. This will encourage us to develop a modified sampling algorithm (MRA). Under the additional constraint of discrete marginals the MRA yields an approximation of the unique tree dependent distribution.
\\ \\
It is, however, questionable if the conditional independence assumption is indeed a reasonable assumption in practice. If it is not, then determining the joint distribution approximated by a sampling algorithm is pointless. Even if we can determine it, it will hardly be an adequate representation of the true joint distribution of the individual risks. \\
Instead, it would be more advisable to gain awareness of the space of mildly tree dependent distributions and how it is affected by the different parameters of the aggregation tree model and the tree structure. An interesting question, for instance, would be whether it is possible to narrow down the number of mildly tree dependent distributions by altering the order in which the risks are aggregated. We will address such questions in the second part of the thesis, where we study and determine the space of mildly tree dependent distributions for simple aggregation tree models. We believe that any insurer or reinsurer using such tools should be aware of the insights gained from this and that this awareness should strongly call for designing trees that adequately fit the business.

\subparagraph{Organisation of the thesis.} Chapter \ref{cha:2} introduces the aggregation tree model with all the basic definitions and results. The sampling algorithm proposed by Arbenz et al. \cite{arbenz12} and its convergence properties are described in Chapter \ref{cha:3}. In Chapter \ref{cha:4} we propose the modified sampling algorithm (MRA) and discuss its convergence. The space of mildly tree dependent distributions is studied in Chapter \ref{chap:spaceofmildly}. Chapter \ref{cha:6} concludes the thesis.
\chapter{Hierarchical risk aggregation}
\label{cha:2}

This chapter introduces the basic definitions and concepts that play an important role throughout this thesis. We give a short introduction to copulas in the first Section \ref{sec:copulas}. In Section \ref{sec:aggtreemodel} we describe step-by-step the hierarchical risk aggregation method. \\
We deliberately choose to follow very closely the structure and terminology proposed by Arbenz et al. \cite{arbenz12}. Readers who are fully familiar with their paper may as well skip this chapter.
\section{Copulas}
\label{sec:copulas}
Copulas have become increasingly popular in the last decades both among academics and practitioners. They turned out to be a powerful and flexible tool to model dependencies between random variables. This is particularly useful in the insurance sector, where the dependence structure between the different risks is of greatest interest. \nocite{cherubini04}
\\ \\
The definition of a copula is surprisingly simple \cite{embrechts05}:
\begin{definition}[Copula]
A $d$-dimensional copula is a distribution function on $[0, 1]^d$
with standard uniform marginal distributions.
\end{definition}
We reserve the notation $C(\boldsymbol{u}) = C(u_1, . . . , u_1 )$ for the multivariate distribution functions that are copulas. Hence, $C$ is a mapping of the form $C : [0, 1]^d \rightarrow [0, 1]$, i.e. a mapping of the unit hypercube into the unit interval.
\\ \\
The importance of copulas is justified by Sklar's Theorem:
\begin{theorem}[Sklar 1959]
\label{the:sklar}
Let $F$ be a joint distribution function with margins
$F_1, \ldots , F_d$. Then there exists a copula $C : [0, 1]^d \rightarrow [0, 1]$ such that, for all
$x_1, \ldots , x_d$ in $\mathbb{R}= [-\infty,\infty]$,
\begin{align}
\label{eq:sklar}
F(x_1, \ldots , x_d ) = C(F_1(x_1), \ldots , F_d(x_d )).
\end{align}
If the margins are continuous, then C is unique. \\
Conversely, if $C$ is a copula and $F_1, \ldots , F_d$ are univariate distribution functions,
then the function $F$ defined in (\ref{eq:sklar})  is a joint distribution function with margins
$F_1, \ldots , F_d$.
\end{theorem}
Sklar's Theorem can be interpreted in the following way:
\begin{enumerate}
\item The first part of the theorem states that we can decompose any distribution function into its margins and a copula. This allows us to study multivariate distributions independently of the margins.
\item The second part of the theorem states that copulas together with margi-nal distribution functions can be used to construct new multivariate distributions.
\end{enumerate}

\section{The aggregation tree model}
\label{sec:aggtreemodel}
In the Introduction we presented the idea of hierarchical risk aggregation, and we learned that this method is graphically best represented by a so called \textit{aggregation tree}. It is our aim to introduce a couple of definitions and notational conventions in order to describe the hierarchical risk aggregation model in an sound mathematical way. We will closely follow the terminology suggested by Arbenz et al. \cite{arbenz12}.

\subsubsection{Rooted tree}
It seems reasonable to adopt part of the terminology from graph theory to describe the tree structure. We will, therefore, first introduce the notion of a \textit{rooted tree}. A rooted tree consists of branching nodes and leaf nodes, with one distinct branching node being the root. The root node is
denoted by $\varnothing$ and the other nodes by tuples $(i_1,\ldots , i_d)$ of natural
numbers $i_j \in \mathbb{N}$. If a node $(i_1,\ldots , i_d) \in \mathbb{N}^d$ is not a leaf node,
it branches into a number of $N_{(i_1,\ldots , i_d)} \in \mathbb{N}$ children, which are
denoted by the $(d + 1)$-tuples $(i_1,\ldots , i_d,1)$, $(i_1,\ldots , i_d, 2), \ldots,
(i_1, \ldots , i_d, N_{(i_1,\ldots , i_d)}) \in \mathbb{N}^{d+1}$.
\begin{definition}
A finite set $\tau \subset \left\lbrace \varnothing\right\rbrace \cup \bigcup_{n=1}^{\infty} \mathbb{N}^n$ denotes a rooted tree if
\begin{enumerate}
\item the root $\varnothing$ is contained in $\tau$,
\item for each node $I=(i_1,\ldots , i_d) \in \tau$, the number of children is given by $N_I \in \mathbb{N}_0$, i.e. the node I has a child $(i_1,\ldots , i_d,k) \in \tau$ if and only if $k \in \left\lbrace n \in \mathbb{N}:1\leq n \leq N_I\right\rbrace$,
\item each node $(i_1,\ldots , i_{d-1},i_d) \in \tau$ has a parent node represented by $(i_1,\ldots , i_{d-1}) \in \tau$.
\end{enumerate}
\end{definition}
\begin{example}
Let $\tau = \left\lbrace \varnothing, ~(1), ~(1,1), ~(1,2), ~(2,1), ~(2,2), ~(2) \right\rbrace$, as illustrated in Figure \ref{fig:4dimrootedtree}. Note that $\tau$ defines a tree of the structure as used in Example \ref{ex:hierarchical}. 
\end{example}
\begin{figure}[h]
\centering
\includegraphics[width=1\linewidth]{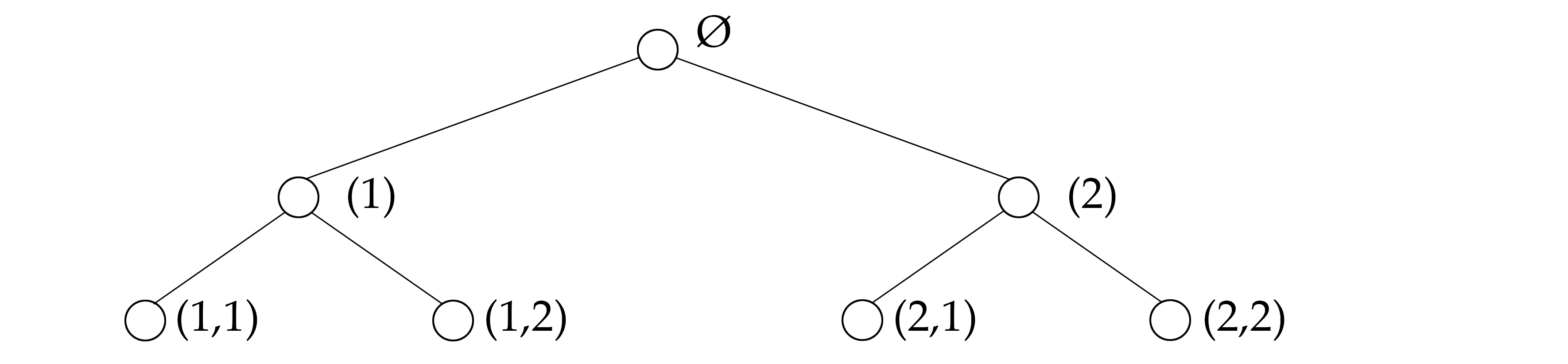}
\caption{An illustration of the rooted tree $\tau = \left\lbrace \varnothing, ~(1), ~(1,1), ~(1,2), ~(2,1), ~(2,2), ~(2) \right\rbrace$.}
\label{fig:4dimrootedtree}
\end{figure}

We define the following subsets for every rooted tree $\tau$:
\begin{enumerate}
\item Leaf nodes are denoted by $\mathscr{L}(\tau)=\left\lbrace I \in \tau : N_I=0\right\rbrace $.
\item Branching nodes are $\mathscr{B}(\tau) = \left\lbrace I \in \tau : N_I > 0 \right\rbrace = \tau \setminus \mathscr{L}(\tau)$.
\item Children of node $I= (i_1,\ldots,i_d) \in \mathscr{B}(\tau)$ are defined as $\mathscr{C}(I,\tau)=\left\lbrace (I,1),\ldots,(I,N_I)\right\rbrace$.
\item Descendants of node  $I= (i_1,\ldots,i_d) \in \tau$ are defined as $\mathscr{D}(I,\tau)=\left\lbrace (j_1,\ldots,j_k) \in \tau : k \geq d, ~(j_1,\ldots,j_d) = (i_1,\ldots,i_d)\right\rbrace$.
\item Leaf descendants of node $I \in \tau$ are defined as $\mathscr{LD}(I,\tau)=\mathscr{L}(\tau)\cap \mathscr{D}(I,\tau)$.
\item The number of leaf descendants of node $I\in \tau$ is denoted with $M_I=\#\mathscr{LD}(I,\tau)$.
\end{enumerate}

\subsubsection{Aggregation tree}

This section introduces a risk aggregation approach building upon a given rooted tree $\tau$. We define on some probability space $(\varOmega, \mathscr{A}, \mathbb{P})$ a random vector $(S_I )_{I \in \tau}$  that assigns to each node $I \in \tau$ a random variable $X_I : \varOmega \rightarrow \mathbb{R}$ such that
\begin{itemize}
\item for leaf nodes $I\in \mathscr{L}(\tau)$, the $X_I$ represent the risks whose aggregate we are interested in,
\item branching nodes $X_I$, $I\in \mathscr{B}(\tau)$, are given by the aggregation of their children: $X_I = X_{I,1}+\ldots+X_{I,N_I}$.
\end{itemize}
In the following, we will further use the convention that whenever we write $\boldsymbol{X}_I$, $I\in \mathscr{B}(\tau)$, in bold letters we mean the \textit{random vector} defined by
\begin{align*}
\boldsymbol{X}_I := (X_J)_{J \in \mathscr{LD}(I,\tau)},
\end{align*}
with components $X_J$, for $J \in \mathscr{LD}(I,\tau)$, where the components are lexicographically ordered with respect to their indices $J$.
\\ \\
For ease of notation, we drop brackets of index vectors as well as the argument $\tau$ whenever the meaning is clear. For instance, $X_{(1,1)} = X_{1,1}$, $\mathscr{D}(I, \tau ) = \mathscr{D}(I)$ and $\mathscr{LD}((1, 1), \tau ) = \mathscr{LD}(1, 1)$.
\\ \\
The rooted tree defined in the previous section represents the hierarchy in which the individual risks are aggregated. In order to complete the aggregation tree model we additionally need to define marginal distributions $F_I$, $I\in \mathscr{L}(\tau)$, and the dependence structure between the children of a branching node $I$. The dependence structure is given through copulas $C_I$, $I \in \mathscr{B}(\tau)$.
\\ \\
In summary, the aggregation tree model is given by the triple
\begin{align}
\label{eq:triple}
\left( \tau,(F_I)_{I\in \mathscr{L}(\tau)},(C)_{I \in \mathscr{B}(\tau)}\right) 
\end{align}
consisting of
\begin{itemize}
\item a rooted tree $\tau$;
\item distribution functions $F_I:\mathbb{R}\rightarrow [0,1]$ for all $I\in \mathscr{L}(\tau)$;
\item copulas $C_I:[0,1]^{N_I} \rightarrow [0,1]$ for all $I \in \mathscr{B}(\tau)$.
\end{itemize}
\begin{example} 
\label{ex:triple}
In our lead Example \ref{ex:hierarchical}, the triple (\ref{eq:triple}) would incorporate the tree $\tau = \left\lbrace \varnothing, ~(1), ~(1,1), ~(1,2), ~(2,1), ~(2,2), ~(2) \right\rbrace$; four univariate distribution functions $F_{1,1}$, $F_{1,2}$, $F_{2,1}$ and $F_{2,2}$; and three bivariate copulas $C_1$, $C_2$ and $C_{\varnothing}$.
\end{example}

\subsubsection{Mildly tree dependent and tree dependent random vectors}
Based on the triple (\ref{eq:triple}), we characterize a class of distributions for a random vector $(X)_{I \in \tau}$.
\begin{definition}
\label{def:mildlytreedependent}
For a given triple (\ref{eq:triple}), $(X)_{I \in \tau}$ is called mildly tree dependent if the following three conditions are met.
\begin{itemize}
\item For each leaf node $I \in \mathscr{L}$, the random variable $X_I$ has a distribution $F_I$; $F_I(x)=P[X_I \leq x]$ for all $x \in \mathbb{R}$.
\item For each branching node $I \in \mathscr{B}$, $X_I$ is the sum of its children, i.e., $X_I = \sum_{i=1}^{N_I}X_{I,i}$. The marginal density function of $X_I$ is denoted by $F_I:\mathbb{R}\rightarrow [0,1]$.
\item For each branching node $I \in \mathscr{B}$, the dependence structure of its children $\mathscr{C}$ is given by the copula $C_I$, i.e.
\begin{align*}
P[X_{I,1}\leq x_1,\ldots,X_{I,N_I}\leq x_{N_I}] = C_I(F_{I,1}(x_1),\ldots,F_{I,N_I}(x_{N_I})),
\end{align*}
for all $(x_1,\ldots,x_{N_I})\in \mathbb{R}^{N_I}$.
\end{itemize}
\end{definition}
The following example illustrates a mildly tree dependent random vector.
\begin{example}
\label{ex:mildlytreedependent}
Consider the aggregation tree model $\left( \tau,(F_I)_{I\in \mathscr{L}(\tau)},(C)_{I \in \mathscr{B}(\tau)}\right)$ from our lead Example \ref{ex:hierarchical}, which was specified in Example \ref{ex:triple}. Suppose $(X_I)_{I \in \tau} = (X_{\varnothing}, ~X_1, ~X_{1,1},~X_{1,2},~X_{2,1},~X_{2,2},~X_2)$ is mildly tree dependent for the model $\left( \tau,(F_I)_{I\in \mathscr{L}(\tau)},(C)_{I \in \mathscr{B}(\tau)}\right)$:
\begin{align*}
P[X_{1,1} \leq x] &= F_{1,1}(x), \hspace{20 pt} P[X_{1,2} \leq x] = F_{1,2}(x), \\ 
P[X_{2,1} \leq x] &= F_{2,1}(x), \hspace{20 pt} P[X_{2,2} \leq x] = F_{2,2}(x),\quad x \in \mathbb{R}.
\end{align*}
Recall that $X_1 = X_{1,1} + X_{1,2}$, $X_2 = X_{2,1} + X_{2,2}$ and $X_{\varnothing} = X_1+X_2$. The copulas $C_1$, $C_2$ and $C_{\varnothing}$ determine the dependence structure of $(X_{1,1},X_{1,2})$, $(X_{2,1},X_{2,2})$ and $(X_1,X_2)$, respectively:
\begin{align*}
P[X_{1,1} \leq x_1,X_{1,2} \leq x_2] &= C_1(F_{1,1}(x_1), F_{1,2}(x_2)), \\
P[X_{2,1} \leq x_1,X_{2,2} \leq x_2] &= C_2(F_{2,1}(x_1), F_{2,2}(x_2)), \\
P[X_1 \leq x_1,X_2 \leq x_2] &= C_{\varnothing}(F_1(x_1), F_2(x_2)).
\end{align*}
Recall that Figure \ref{fig:4dimtree1} illustrates the aggregation structure.
\end{example}
By Sklar's Theorem \ref{eq:sklar} it should be clear that in the above Example \ref{ex:mildlytreedependent} the distribution of the random vectors $(X_{1,1},X_{1,2})$, $(X_{2,1},X_{2,2})$, $(X_1,X_2)$, and hence also of the aggregated risk $X_{\varnothing}=X_1+X_2$, is uniquely specified through the aggregation tree model $( \tau,(F_I)_{I\in \mathscr{L}(\tau)},(C)_{I \in \mathscr{B}(\tau)})$. We will state a general result in the next section. \\
On the other hand, as mentioned already in the Introduction, the joint distribution of all risk $(X_{1,1},X_{1,2},X_{2,1},X_{2,2})$ is not uniquely specified. Put differently: given a aggregation model $( \tau,(F_I)_{I\in \mathscr{L}(\tau)},(C)_{I \in \mathscr{B}(\tau)})$ there exists in general more than one mildly tree dependent random vector $(X)_{I \in \tau}$. \\ \\
Arbenz et al. \cite{arbenz12} propose an additional condition which makes the joint distribution unique. The condition is formulated in the following definition.
\begin{definition}
\label{def:treedependent}
For a given triple (\ref{eq:triple}), a mildly tree dependent random vector $(X)_{I \in \tau}$ is called tree dependent if for each branching node $I \in \mathscr{B}(\tau)$, given $X_I$, its descendants $(X_J)_{J \in \mathscr{D}(I)}$ are conditionally independent of the remaining nodes $(X_J)_{J \in \tau \setminus \mathscr{D}(I)}$:
\begin{align}
\label{def:CIA}
(X_J)_{J \in \mathscr{D}(I)} \bot (X_J)_{J \in \tau \setminus \mathscr{D}(I)}\mid X_I \qquad \text{for all } I \in \mathscr{B}(\tau). 
\end{align}
In the sequel, (\ref{def:CIA}) will be referred to as the conditional independence assumption.
\end{definition}

\begin{example}
Considering once again the aggregation model from our previous examples, the conditional independence assumption would read:
\begin{align*}
(X_{1,1},X_{1,2},X_1) \bot (X_{2,1},X_{2,2},X_2,X_{\varnothing}) \mid X_1, \\
(X_{2,1},X_{2,2},X_2) \bot (X_{1,1},X_{1,2},X_1,X_{\varnothing}) \mid X_2.
\end{align*}
One can show that under this additional condition, a tree dependent random vector exists and is unique. We will state the general result in the next section.
\end{example}
The conditional independence assumption (\ref{def:CIA}) may seem quite confusing at first glance, and it is questionable whether it is a reasonable assumption in practice or not. We will come back to that later on in Section \ref{sec:CIA}.

\section{Existence and uniqueness}
In this section we address existence and uniqueness of the distribution of tree dependent and mildly tree dependent random vectors. \\ \\
Recalling again the definitions of tree dependence and mild tree dependence, it is clear that mild tree dependence covers a larger space of multivariate distributions than tree dependence. Any tree dependent random vector is mildly tree dependent, whereas the converse does not hold. \\ \\
Given an aggregation tree model $\left( \tau,(F_I)_{I\in \mathscr{L}(\tau)},(C)_{I \in \mathscr{B}(\tau)}\right)$, the existence of a mildly tree dependent random vector $(X)_{I \in \tau}$ is fairly obvious by Sklar's Theorem \ref{eq:sklar}. And we have already mentioned before that in general the distribution of a mildly tree dependent random vector is not unique. Considering the smaller class of tree dependent random vectors existence is non-obvious. Arbenz et al. \cite{arbenz12} proved that among the mildly tree dependent distributions there exists indeed exactly one distribution that satisfies (\ref{def:CIA}): 
\begin{theorem}
Given an aggregation tree model (\ref{eq:triple}), a tree dependent random vector $(X)_{I \in \tau}$ exists and its joint distribution is unique.
\end{theorem}
\begin{proof}
See \cite{arbenz12} p. 125.
\end{proof}
We also state the following result regarding an important property of mildly tree dependent random vectors.
\begin{proposition}
\label{prop:mild}
Given an aggregation tree model (\ref{eq:triple}) and some fixed $I\in\mathscr{B}$. Then
the distribution of the vector $(X_{I,1},\ldots, X_{I,N_I} )$ is equal for all mildly
tree dependent vectors $(X)_{I \in \tau}$.
\end{proposition}
\begin{proof}
Simple application of Sklar's Theorem \ref{eq:sklar}. See \cite{arbenz12} p. 126.
\end{proof}
In particular, by Proposition \ref{prop:mild} the distribution of the total aggregate $X_{\varnothing}$ is the same for any mildly tree dependent random vector $(X)_{I \in \tau}$.

\chapter{The sample reordering algorithm}
\label{cha:3}
As in most cases the quantities of interest (e.g., the distribution of $X_{\varnothing}$) of an aggregation tree model cannot be calculated analytically, Arbenz et al. \cite{arbenz12} propose a sampling algorithm for the numerical approximation of mild tree dependence. The samples obtained by the algorithm can then, for instance, be used to approximate the distribution of the total aggregate $X_{\varnothing}$.
\\ \\
In this chapter we first present the sample reordering algorithm proposed by Arbenz et al. \cite{arbenz12}. We then state some convergence results and conduct a numerical example at the end of this chapter that raises new questions and motivates the subsequent chapters.
\section{Definition of the algorithm}
The sampling algorithm by Arbenz et al. \cite{arbenz12} uses a bottom-up approach to achieve a numerical approximation of mild tree dependence. The basic idea can be summarized as follows:
\begin{enumerate}
\item Simulate independent marginal samples for all leaf nodes $X_I$, $I \in \mathscr{L}$;
\item simulate independent copula samples from $C_I$ for all branching nodes $I \in \mathscr{B}$;
\item for $I \in \mathscr{B}$, recursively define approximations of the distributions of $X_I$ and $(X_{I,1},\ldots, X_{I,N_I} )$ based on the empirical margins and the empirical copulas from steps 1 to 2.
\end{enumerate}
For ease of notation, let 
$G_I : \mathbb{R}^{N_I} \rightarrow [0, 1]$ for $I \in \mathscr{B}$ denote the
joint distribution function of $(X_{I,1},\ldots, X_{I,N_I} )$, the children of $X_I$:
\begin{align*}
G_I(x_1,\ldots,x_{N_I}) &= P[X_{I,1}\leq x_1,\ldots,X_{I,N_I}\leq x_{N_I}] \\
&=C_I(F_{I,1}(x_1),\ldots,F_{I,N_I}(x_{N_I})).
\end{align*}
\begin{algorithm}
\label{alg:arbenz} Fix $n \in \mathbb{N}$
\begin{enumerate}
\item Generate $n$ independent samples from the leaf nodes $X_I$, $I \in \mathscr{L}$,
and the copulas $C_I$, $I \in \mathscr{B}$
\begin{itemize}
\item $X_I^k \sim F_I$, for $k=1,\ldots,n$,
\item $\boldsymbol{U}_I^k = (U_I^{k,1},\ldots,U_I^{k,N_I}) \sim C_I$, for $k=1,\ldots,n$.
\end{itemize}
\item Define empirical margins $F_I^n:\mathbb{R} \rightarrow [0,1]$, $I \in \mathscr{L}$, by 
\begin{align*}
F_I^n(x) = \frac{1}{n} \sum\limits_{k=1}^{n}\mathbb{1}_{\left\lbrace X_I^k \leq x\right\rbrace },~x \in \mathbb{R},
\end{align*}
and empirical copulas $C_I^n:[0,1]^{N_I} \rightarrow [0,1]$, $I \in \mathscr{B}$, by
\begin{align*}
C_I^n(\boldsymbol{u}) = \frac{1}{n} \sum\limits_{k=1}^{n}\mathbb{1}_{\left\lbrace \frac{R_I^{k,1}}{n}\leq u_1,\ldots,\frac{R_I^{k,N_I}}{n}\leq u_{N_1}\right\rbrace }, \text{ for } \boldsymbol{u} = (u_1,\ldots,u_{N_I}) \in [0,1]^{N_I},
\end{align*}
where $R_I^{k,i}$ is the rank of $U_I^{k,i}$ within the set $\left\lbrace U_I^{j,i}\right\rbrace _{j=1}^n$:
\begin{align*}
R_I^{k,i} = \sum\limits_{j=1}^{n}\mathbb{1}_{\left\lbrace U_I^{j,i} \leq U_I^{k,i} \right\rbrace }.
\end{align*}
\item For $I \in \mathscr{B}$, recursively define as approximations of $G_I$ and $F_I$:
\begin{align*}
G_I^n(x_1,\ldots,x_{N_I}) = C_I^n(F_{I,1}^n(x_1),\ldots,F_{I,N_I}^n(x_{N_I})),
\end{align*}
for $(x_1,\ldots,x_{N_I}) \in \mathbb{R}^d$ and
\begin{align*}
F_I^n(t) = \int_{\mathbb{R}^{N_I}} \mathbb{1}{\left\lbrace x_1+\ldots+x_{N_I}\leq t\right\rbrace } dG_I^n(x_1,\ldots,x_{N_I}),
\end{align*}
for $t \in \mathbb{R}$.
\end{enumerate}
\end{algorithm}
It turns out that applying empirical copulas to empirical margins, as used in the definition of the $G_I^n$, can be efficiently represented in terms of sample reordering. The idea goes back to Iman and Conover \cite{iman82} and was adapted by Arbenz et al. \cite{arbenz12} to this context: 
\begin{theorem}
\label{the:samplereordering}
In the following, a permutation denotes a bijective mapping from $\left\lbrace 1, 2, . . . , n\right\rbrace $ to $\left\lbrace 1, 2, . . . , n\right\rbrace $.
\begin{itemize}
\item Let the permutations $p_I$ for $I \in \tau \setminus \varnothing$ be defined through $p_{I,i}(k)=R_I^{k,i}$, $k=1,\ldots,n$.
\item Recursively define samples $X_I^k$, $k=1,\ldots,n$, for $I \in \mathscr{B}$ by
\begin{align*}
X_I^k = \sum_{J \in \mathscr{C}(I)} X_J^{(p_J(k))} = X_{I,1}^{(p_{I,1}(k))} + \ldots + X_{I,N_I}^{(p_{I,N_I}(k))},
\end{align*}
where $X_J^{(k)}$ denotes the k-th order statistic of $\left\lbrace X_J^1,\ldots,X_J^n\right\rbrace$:
\begin{align*}
X_J^{(1)} \leq X_J^{(2)} \ldots \leq X_J^{(n)}.
\end{align*}
\end{itemize}
Then $G_I^n$, for $I \in \mathscr{B}$, and $F_I^n$, for $I \in \tau$, satisfy
\begin{align}
\label{eq:theorem_convergenceArbenz1}
&G_I^n(x_1,\ldots,x_{N_I})= \frac{1}{n} \sum\limits_{k=1}^{n}\mathbb{1}_{\left\lbrace  X_{I,i}^{(p_{I,i}(k))} \leq x_i \text{ for all } i=1,\dots,N_I \right\rbrace }, \\
\label{eq:theorem_convergenceArbenz2}
&F_I^n(x) = \frac{1}{n} \sum\limits_{k=1}^{n}\mathbb{1}_{\left\lbrace X_I^k \leq x\right\rbrace }.
\end{align}
\end{theorem}
\begin{proof}
See \cite{arbenz12} p. 127.
\end{proof}
Algorithm \ref{alg:arbenz} can now be formulated in terms of sample reorderings. We start with simulating leaf nodes and copulas independently, thus obtaining samples $X_I^k$, $k = 1,\ldots , n$, for $I \in \mathscr{L}$ and $\boldsymbol{U}_I^k$, $k = 1,\ldots, n$, for $I \in \mathscr{B}$. We then use the permutations $p_{I,i}$ to link the appropriate order statistics and define the atoms of $G_I^n$ by
\begin{align}
\label{eq:atom}
(X_{I,1}^{(p_{I,1}(k))}, \ldots ,X_{I,N_I}^{(p_{I,N_I}(k))}), \hspace{15pt} k=1,\ldots,n
\end{align}
Hence, the permutations $p_I$ can be used to introduce the correct dependence structure into originally independent samples by joining the appropriate order statistics. For branching nodes $I \in \mathscr{B}$, the component-wise sum then defines the atoms of $F_I$ as $X_I^k=X_{I,1}^{(p_{I,1}(k))} + \ldots + X_{I,N_I}^{(p_{I,N_I}(k))}$. We iteratively repeat this reordering from the bottom to the top of the tree. An illustrating example can be found in \cite{arbenz12}.\\ \\
We will later on propose an alternative way to reorder the samples. The alternative way seems less natural and is notationally  less comprehensible but yields exactly the same atoms, so that the statement of Theorem \ref{the:samplereordering} still holds true. In addition, it allows us to draw some conclusions regarding the joint distribution of the  mildly tree dependent vector approximated by our modified algorithm. Having said that, we go on by discussing some convergence results in the next section.

\section{Convergence result}
In this section we state the main convergence result on Algorithm \ref{alg:arbenz}. We restrain ourselves, at this place, to provide an in-depth deduction and analysis of the following result, and we refer to Mainik \cite{mainik11} for a detailed deduction in a much more general setting. \\ \\
In order to formulate the results, we need some additional notation. For continuous margins and $t \in \mathbb{R}$ define $B_I(t) \in [0, 1]^{N_I}$ as
\begin{align*}
B_I(t)=\left\lbrace (F_{I,1}(x_1),\ldots,F_{I,N_I}(x_{N_I})) : (x_1,\ldots,x_{N_I}) \in \mathbb{R}^{N_I}, x_1+\ldots+x_{N_I}=t \right\rbrace 
\end{align*}
Let $U_{\delta}(B_I(t))$ denote the $\delta$-neighbourhood of $B_I(t)$ in $[0, 1]^{N_I}$:
\begin{align*}
U_{\delta}(B_I(t))=\left\lbrace \boldsymbol{x} \in [0, 1]^{N_I} : \lVert \boldsymbol{x}-\boldsymbol{y} \rVert_2 < \delta \text{ for some } \boldsymbol{y} \in B_I(t) \right\rbrace.
\end{align*}
For absolutely continuous copulas $C_I$, we denote the density by
$c_I : [0, 1]^{N_I} \rightarrow [0,\infty)$.
\begin{theorem}
\label{the:convergencearbenz}
Assume that for all $I \in \mathscr{B}$ the copulas $C_I$ are absolutely
continuous and satisfy
\begin{align*}
\lim\limits_{\delta \rightarrow \infty} \sup\limits_{t \in \mathbb{R}} \int_{U_{\delta}(B_I(t))} c_I(u_1,\ldots,u_{N_I})du_1\cdots du_{N_I} = 0.
\end{align*}
Then, for each branching node $I \in \mathscr{B}$,
\begin{align*}
&\lim\limits_{n \rightarrow \infty} \sup\limits_{\boldsymbol{x} \in \mathbb{R}^{N_I}} \left| G_I^n(\boldsymbol{x}) - G_I(\boldsymbol{x})\right|  = 0 \quad P-a.s., \\
&\lim\limits_{n \rightarrow \infty} \sup\limits_{t \in \mathbb{R}} \left| F_I^n(t) - F_I(t)\right|  = 0 \quad P-a.s.
\end{align*}
\end{theorem}
\begin{proof}
See Mainik \cite{mainik11}.
\end{proof}
It is crucial to grasp the meaning of Theorem \ref{the:convergencearbenz} correctly. Theorem \ref{the:convergencearbenz} does not make any statement whether the reordered samples from Algorithm \ref{alg:arbenz} approximate a (mildly) tree dependent random vector. Instead, it simply tells us that as long as the conditions on the copulas $C_I$ are satisfied, the samples approximate the distributions $G_I$, $I \in \mathscr{B}$, which characterize all mildly tree dependent distributions. In particular, the samples approximate the distribution $F_{\varnothing}$ of the total aggregate $X_{\varnothing}$.

\section{Numerical experiment}
\label{sec:numericalExp}
In the previous sections we introduced the hierarchical aggregation tree model. So far, papers on this subject focused on the total aggregate $X_{\varnothing}$ as the quantity of interest. Arbenz et al. \cite{arbenz12} proposed the reordering Algorithm \ref{alg:arbenz} (incl. Theorem \ref{the:samplereordering}), which yields samples that approximate the distribution of the total aggregate (recall Theorem \ref{the:convergencearbenz}). Unaddressed remained the question whether the samples actually approximate a mildly tree dependent or even tree dependent random vector. 
\\ \\
In the following we will now shift our focus towards this relevant question. We emphasize that what follows has not yet been addressed in any previous paper and goes beyond the current theoretical results on this subject. \\ \\
We start the discussion with a numerical motivational example. The result of the example is surprising and gives a first hint what the answer to our question might look like.
\subsubsection{Defining the model}
To start with, the aggregation tree model $\left( \tau,(F_I)_{I\in \mathscr{L}(\tau)},(C)_{I \in \mathscr{B}(\tau)}\right)$ needs to be specified. We consider once again the situation from our lead Example \ref{ex:hierarchical}. By Example \ref{ex:triple}, we know that in this set up the aggregation tree model incorporates the tree $\tau = \left\lbrace \varnothing, ~(1), ~(1,1), ~(1,2), ~(2,1), ~(2,2), ~(2) \right\rbrace$; four univariate density functions $F_{1,1}$, $F_{1,2}$, $F_{2,1}$ and $F_{2,2}$; and three bivariate copulas $C_1$, $C_2$ and $C_{\varnothing}$.
\\ \\
Suppose $F_{1,1}$, $F_{1,2}$, $F_{2,1}$ and $F_{2,2}$ are given by normal distributions with means 4,2,0,3; and variances 3, 4, 10, 2, respectively:
\begin{align*}
F_{1,1}(x) = \varPhi\left( \frac{x-4}{\sqrt{3}}\right), \hspace{25 pt} F_{1,2}(x) = \varPhi\left( \frac{x-2}{\sqrt{4}}\right), \\
F_{2,1}(x) = \varPhi\left( \frac{x}{\sqrt{10}}\right), \hspace{25 pt} F_{2,2}(x) = \varPhi\left( \frac{x-3}{\sqrt{2}}\right).
\end{align*}
Suppose further that $C_1$, $C_2$ and $C_{\varnothing}$ are normal copulas with correlation matrices
\begin{align*}
R_1 = \left(  \begin{array}{cc}
1 & 0.7 \\
0.7 & 1
\end{array}\right), \hspace{25 pt}
R_2 = \left(  \begin{array}{cc}
1 & 0.5 \\
0.5 & 1
\end{array}\right), \hspace{25 pt}
R_{\varnothing} = \left(  \begin{array}{cc}
1 & 0.2 \\
0.2 & 1
\end{array}\right).
\end{align*}
Note that in our model all leaf nodes are normal, and all copulas are normal copulas. Hence, we are dealing with a Gaussian aggregation tree model. Gaussian aggregation tree models have the convenient property that the unique tree dependent vector $(X_I)_{I \in \tau}$ is multivariate normal distributed and can be explicitly calculated by a simple recursion formula (see Proposition 2.8 in \cite{arbenz12}). \\ 
For the tree dependent random vector $(X_I)_{I \in \tau}$, we deduce by Proposition 2.8 in \cite{arbenz12} that the distribution of $\boldsymbol{X}_{\varnothing} = (X_J)_{J \in \mathscr{LD}(\varnothing)} = (X_{1,1},X_{1,2},X_{2,1},X_{2,2}) \sim \mathcal{N}(\mu_{\varnothing}, \varSigma_{\varnothing})$ is multivariate normal with 
\begin{align}
\label{eq:covarianceTreeDep}
\mu_{\varnothing}=\left( \begin{array}{c}
4 \\
2 \\
0 \\
3
\end{array}\right) , \hspace{30 pt}
\varSigma_{\varnothing} = \left( \begin{array}{cccc}
3 & 2.4249 & 0.9502 & 0.3290 \\
2.4249 & 4 & 1.1254 & 0.3896 \\
0.9502 & 1.1254 & 10 & 2.2361 \\
0.3290 & 0.3896 & 2.2361 & 2
\end{array}\right). 
\end{align}
\subsubsection{Numerical result}
We will now apply the reordering algorithm on our aggregation tree model specified above. Let $(X_I)_{I \in \tau}$ be the tree dependent random vector and note that by the convergence result stated in Theorem \ref{the:samplereordering} we know that the reordered samples approximate the distributions of
\begin{itemize}
\item $\boldsymbol{X}_1=(X_{1,1}, X_{1,2})$ and $\boldsymbol{X}_2=(X_{2,1}, X_{2,2})$;
\item $X_1=X_{1,1}+ X_{1,2}$ and $X_2=X_{2,1}+ X_{2,2}$;
\item $(X_{1}, X_{2})$ and the total aggregate $X_{\varnothing} = X_{1}+ X_{2}$.
\end{itemize}
Still unclear is whether the samples provide an approximation of the tree dependent random vector characterized by the distribution of $\boldsymbol{X}_{\varnothing} \sim \mathcal{N}(\mu_{\varnothing}, \varSigma_{\varnothing})$. \\ \\
In order to investigate this question, we compute the sample covariance matrix $\bar{\varSigma}_{\varnothing}$ of the reordered samples for a sample size of $n=10^7$:
\begin{align}
\bar{\varSigma}_{\varnothing} = \left( \begin{array}{cccc}
2.9985 & 2.4252 & 0.9513 & 0.3301 \\
2.4252 & 4.0025 & 1.1278 & 0.3909 \\
0.9513 & 1.1278 & 9.9978 & 2.2337 \\
0.3301 & 0.3909 & 2.2337 & 1.9994
\end{array}\right). 
\end{align}
It is remarkable that the sample covariance $\bar{\varSigma}_{\varnothing}$ of the reordered samples and the covariance matrix $\varSigma_{\varnothing}$ in (\ref{eq:covarianceTreeDep}) of the tree dependent random vector are almost identical. This is a strong indication that in our example the reordering algorithm yields samples that approximate the unique tree dependent random vector. \\ \\
To back up our hypothesis, we also conduct a multivariate normality test. Recall that by Proposition 2.8 in \cite{arbenz12} the tree dependent random vector is multivariate normally distributed. There are many test for assessing multivariate normality, we choose to apply "Henze-Zirkler's Multivariate Normality Test" \cite{normtest07}. At the 0.05 significance level the test does not reject the null-hypothesis of multivariate normality, and the p-value associated with the Henze-Zirkler statistic is 0.7848365. \\ \\
The numerical results suggest that the reordered samples approximate the unique tree dependent random vector. Naturally, this begs the question whether this holds true in general or everything was just coincidence. We will address this question in the next chapter.

\subsubsection{Concluding remark} \label{sec:concludingRemark}
One might wonder if the above questions are actually relevant. Do we really need to bother ourselves with characterizing the multivariate distribution approximated by the reordered samples? Is it not enough to know that the total aggregate $X_{\varnothing}$ is approximated correctly? \\
Consider, for instance, the above example and recall that in Example \ref{ex:hierarchical} we said that $X_{1,1}$ stands for "Car insurance Switzerland", $X_{1,2}$ stands for "Car insurance Italy", $X_{2,1}$ stands for "Earthquake Switzerland" and $X_{2,2}$ stands for "Earthquake Italy". Assume an insurance company would not be just interested in the total aggregate but also in the total risk they are taking in Switzerland, i.e. $X_{1,1}+X_{2,1}$. This distribution is not uniquely specified by the aggregation tree model illustrated in Figure \ref{fig:4dimtree1}, and so we most likely wish to have certainty about which distribution is approximated by a sampling algorithm. We see that it is definitely worthwhile to clarify the questions that have arisen in this section.

\chapter{Tree dependent sampling}
\label{cha:4}
We have seen in the end of the previous chapter that some issues regarding the reordering algorithm have not yet been addressed sufficiently. In particular, it is not clear whether the samples approximate the unique tree dependent distribution or some other distribution that remains to be determined. The numerical experiment in Section \ref{sec:numericalExp} suggests that indeed the  tree dependent distribution is approximated. \\ \\
In this chapter we propose a modification of the reordering algorithm for which we claim that for discrete marginals we obtain approximations of the tree dependent distribution. In addition, our modified algorithm has the convenient property that it generates i.i.d. realizations; however, at the cost of a worse algorithmic efficiency. \\
We begin the chapter by presenting the modified algorithm. The proof of the above claim will then be the subject of the second part of this chapter.

\section{A modified reordering algorithm}
Based on Arbenz' reordering algorithm (Algorithm \ref{alg:arbenz} \& Theorem \ref{the:samplereordering}) we propose a modified reordering algorithm (MRA). The modifications will play an important role in the proof of Theorem \ref{the:conv2}, which we state in the next section. 
\\ \\
The first modification concerns the order the realizations are linked to each other. Have a look at the following Example \ref{ex:altReordering} where we illustrate what we mean by this.

\begin{example}
\label{ex:altReordering}
Consider the simplest tree $\tau = \left\lbrace \varnothing,~1,~2 \right\rbrace$, with two margins $X_1 \sim F_1$, $X_2 \sim F_2$ and a dependence structure given by a copula $\boldsymbol{U}_{\varnothing} \sim C_{\varnothing}$. Suppose the simulation for $n=3$ yields the following $X_1^k$, $X_2^k$ and $\boldsymbol{U}_{\varnothing}^k$:
\begin{align*}
&X_1^1 = 1, \hspace{20 pt} X_1^2 = 4, \hspace{20 pt} X_1^3 = 2, \\
&X_2^1 = 9, \hspace{20 pt} X_2^2 = 0, \hspace{20 pt} X_2^3 = 3, \\
&\boldsymbol{U}_{\varnothing}^1 = (0.6,0.8), \hspace{20 pt} \boldsymbol{U}_{\varnothing}^2 = (0.3,0.7), \hspace{20 pt} \boldsymbol{U}_{\varnothing}^3 = (0.5,0.1).
\end{align*}
Recall that we denote by $R_I^{k,i}$ the rank of $U_I^{k,i}$ within the set $\left\lbrace U_I^{j,i}\right\rbrace _{j=1}^n$. The ranks of the copula samples are
\begin{align*}
(R_{\varnothing}^{1,1},R_{\varnothing}^{1,2}) = (3,3), \hspace{20 pt} (R_{\varnothing}^{2,1},R_{\varnothing}^{2,2}) = (1,2), \hspace{20 pt} (R_{\varnothing}^{3,1},R_{\varnothing}^{3,2}) = (2,1).
\end{align*}
With $p_{I,i}(k):=R_I^{k,i}$ the permutations are
\begin{align*}
p_{\varnothing,1} : (1,2,3) \mapsto (3,1,2), \hspace{20 pt} p_{\varnothing,2} : (1,2,3) \mapsto (3,2,1).
\end{align*}
Furthermore, we denote by $Q_{\varnothing}^{k,1}$  (resp. $Q_{\varnothing}^{k,2}$) the rank of $X_1^k$ (resp. $X_2^k$) within the set $\left\lbrace X_1^k \right\rbrace_{k=1}^3$ (resp. $\left\lbrace X_2^k \right\rbrace_{k=1}^3$) and define the permutations $q_{\varnothing,1}(k):=Q_{\varnothing}^{k,1}$ and $q_{\varnothing,2}(k):=Q_{\varnothing}^{k,2}$:
\begin{align*}
q_{\varnothing,1} : (1,2,3) \mapsto (1,3,2), \hspace{20 pt} q_{\varnothing,2} : (1,2,3) \mapsto (3,1,2).
\end{align*}
We will also use the notation of the inverse permutations $p_{\varnothing,1}^{-1}$, $p_{\varnothing,2}^{-1}$, $q_{\varnothing,1}^{-1}$ and $q_{\varnothing,2}^{-1}$. 
\\ \\
In the following we present two different reordering orders. We start with the more natural one proposed by Arbenz et al. \cite{arbenz12} in Theorem \ref{the:samplereordering}.
\paragraph{Reordering 1 (Arbenz)} Let $_{Re1}\boldsymbol{X}_{\varnothing}^1$, $_{Re1}\boldsymbol{X}_{\varnothing}^2$, $_{Re1}\boldsymbol{X}_{\varnothing}^3$ denote the reordered vectors. The subscript "\textit{Re1}" indicates that "Reordering 1" was applied. Following the reordering order from Theorem \ref{the:samplereordering} we get
\begin{align*}
_{Re1}\boldsymbol{X}_{\varnothing}^k := \left( \begin{array}{c} X_1^{(p_{\varnothing,1}(k))} \\
X_2^{(p_{\varnothing,2}(k))} 
\end{array}
\right) =  \left( \begin{array}{c} X_1^{q_{\varnothing,1}^{-1}(p_{\varnothing,1}(k))} \\
X_2^{q_{\varnothing,2}^{-1}(p_{\varnothing,2}(k))} 
\end{array}
\right) \hspace{20 pt} \text{for } k=1,\ldots,3.
\end{align*}
The figure below illustrates how the samples are linked to each other ($X_1^k$ top row, $\boldsymbol{U}_{\varnothing}^k$ middle row, $X_2$ bottom row):
\begin{figure*}[h]
\centering
\includegraphics[width=1.1\linewidth]{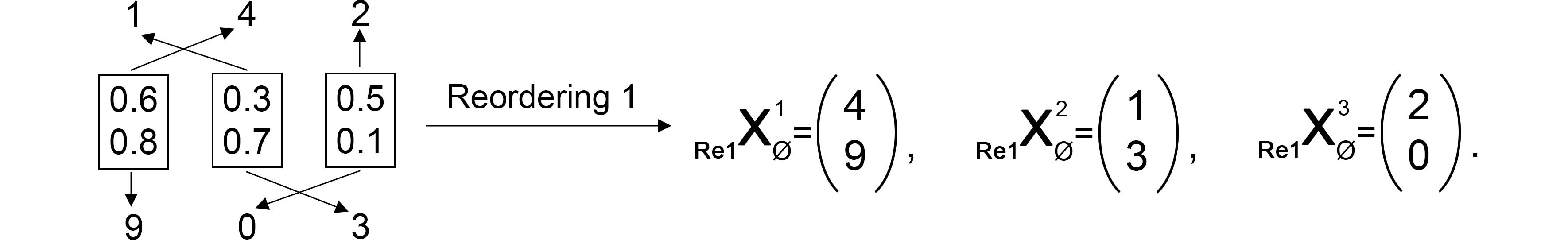}
\label{fig:reordering1}
\end{figure*}
\paragraph{Reordering 2} In the above "Reordering 1" we linked the appropriate order statistics to the ranks associated with the components of $\boldsymbol{U}_{\varnothing}^k$ in order to obtain $_{Re1}\boldsymbol{X}_{\varnothing}^k$. This procedure is illustrated by the arrows pointing from $\boldsymbol{U}_{\varnothing}^k$ towards the appropriate marginals. \\
An alternative way to reorder the samples would be to set the first component of $_{Re2}\boldsymbol{X}_{\varnothing}^k$ equal to $X_1^k$ and then link the appropriate second component to it: 
\begin{figure*}[h]
\centering
\includegraphics[width=1.1\linewidth]{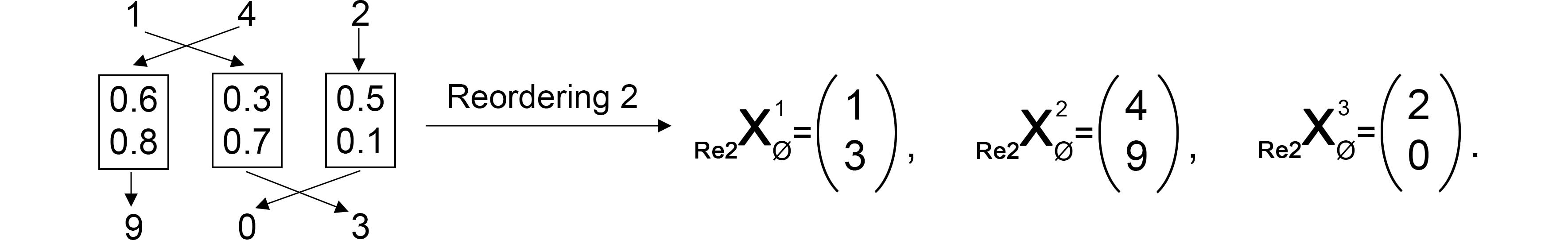}
\label{fig:reordering2}
\end{figure*}
\\ Note how we changed the direction of the arrows in order to illustrate the difference between the two reordering orders. Mathematically, $_{Re2}\boldsymbol{X}_{\varnothing}^k$ can be written as
\begin{align}
\label{eq:re2}
_{Re2}\boldsymbol{X}_{\varnothing}^k := \left( \begin{array}{c} X_1^k \\
X_2^{q_{\varnothing,2}^{-1}(p_{\varnothing,2}(p_{\varnothing,1}^{-1}(q_1(k)))} 
\end{array}
\right) \hspace{20 pt} \text{for } k=1,\ldots,3.
\end{align} 
\end{example}
It is clear that both reordering orders yield the same atoms, i.e. the sets $\left\lbrace _{Re1}\boldsymbol{X}_{\varnothing}^1, _{Re1}\boldsymbol{X}_{\varnothing}^2, _{Re1}\boldsymbol{X}_{\varnothing}^3\right\rbrace $ and $\left\lbrace _{Re2}\boldsymbol{X}_{\varnothing}^1, _{Re2}\boldsymbol{X}_{\varnothing}^2, _{Re2}\boldsymbol{X}_{\varnothing}^3\right\rbrace $ contain the same elements but in general in a different order. As a consequence, the "modification" has in fact absolutely no impact on the convergence result in Theorem \ref{the:convergencearbenz}.
\\ \\
The second modification may be seen as the major modification. It aims at changing the algorithm in such a way that it yields i.i.d. realizations. Let us focus for a moment on the random vectors $_{Re2}\boldsymbol{X}_{\varnothing}^1, _{Re2}\boldsymbol{X}_{\varnothing}^2, _{Re2}\boldsymbol{X}_{\varnothing}^3$ in (\ref{eq:re2}) and note that these are indeed random vectors consisting of the i.i.d. random variables $X_1^1 ,X_1^2,X_1^3 \sim F_1$ and $X_2^1 ,X_2^2,X_2^3 \sim F_2$. By a simple symmetry argument it is clear that $_{Re2}\boldsymbol{X}_{\varnothing}^1, _{Re2}\boldsymbol{X}_{\varnothing}^2, _{Re2}\boldsymbol{X}_{\varnothing}^3$ are identically distributed. They are, however, as Arbenz et al. pointed out already correctly, not independent, and hence the reordered random vectors $_{Re2}\boldsymbol{X}_{\varnothing}^1, _{Re2}\boldsymbol{X}_{\varnothing}^2, _{Re2}\boldsymbol{X}_{\varnothing}^3$ are not i.i.d.. 
\\ \\
To get around this limitation we could proceed as follows: Conduct the algorithm three times in order to obtain three independent copies of the reordered samples. We denote the first copy by  $\left\lbrace ^{\ \ \ 1}_{Re2}\boldsymbol{X}_{\varnothing}^1, {^{\ \ \ 1}_{Re2}\boldsymbol{X}}_{\varnothing}^2, {^{\ \ \ 1}_{Re2}\boldsymbol{X}}_{\varnothing}^3\right\rbrace$, the second one by $\left\lbrace ^{\ \ \ 2}_{Re2}\boldsymbol{X}_{\varnothing}^1, {^{\ \ \ 2}_{Re2}\boldsymbol{X}}_{\varnothing}^2, {^{\ \ \ 2}_{Re2}\boldsymbol{X}}_{\varnothing}^3\right\rbrace$, and the third one by $\left\lbrace ^{\ \ \ 3}_{Re2}\boldsymbol{X}_{\varnothing}^1,{^{\ \ \ 3}_{Re2}\boldsymbol{X}}_{\varnothing}^2, \right. $ $ \left. {^{\ \ \ 3}_{Re2}\boldsymbol{X}}_{\varnothing}^3\right\rbrace$ and set
\begin{align}
\label{eq:diagonalize}
_{Re2}\boldsymbol{X}_{\varnothing}^1 := {^{\ \ \ 1}_{Re2}\boldsymbol{X}}_{\varnothing}^1, \hspace{20pt} _{Re2}\boldsymbol{X}_{\varnothing}^2 := {^{\ \ \ 2}_{Re2}\boldsymbol{X}}_{\varnothing}^2, \hspace{20pt} _{Re2}\boldsymbol{X}_{\varnothing}^3 := {^{\ \ \ 3}_{Re2}\boldsymbol{X}}_{\varnothing}^3.
\end{align}
The so defined reordered random vectors $_{Re2}\boldsymbol{X}_{\varnothing}^1, _{Re2}\boldsymbol{X}_{\varnothing}^2, _{Re2}\boldsymbol{X}_{\varnothing}^3$ are now clearly i.i.d.. 
Note, however, that this comes at the cost of a worse algorithmic efficiency.
\\ \\
Having discussed the two major modifications, we are now ready to formulate and understand the MRA. Unfortunately, the algorithm becomes notationally extremely complex due to the mentioned modifications. Instead of struggling with the notation, we suggest to rather focus on the general idea presented in the previous discussion and to have a look at the explanatory Remark \ref{re:MRA}. \\ \\
In the following, a permutation denotes a bijective mapping from $\left\lbrace 1, 2, . . . , n\right\rbrace $ to $\left\lbrace 1, 2, . . . , n\right\rbrace $.
\begin{algorithm}[Modified reordering algorithm - MRA]
\label{the:MRA} Fix $n \in \mathbb{N}$. \\
The samples are recursively defined from the bottom to the top of the tree.
\begin{enumerate}
\item Generate $n\times n$ independent samples from the leaf nodes $X_I$, $I \in \mathscr{L}$,
and the copulas $C_I$, $I \in \mathscr{B}$
\begin{itemize}
\item $^{\ell}X_I^k \sim F_I$, for $(\ell,k) \in \left\lbrace 1,\ldots,n\right\rbrace ^2$,
\item $^{\ell}\boldsymbol{U}_I^k = (^{\ell}U_I^{k,1},\ldots,{^{\ell}U}_I^{k,N_I}) \sim C_I$, for $(\ell,k) \in \left\lbrace 1,\ldots,n\right\rbrace ^2$.
\end{itemize}
\end{enumerate}
Recall that we denote by $^{\ell}R_I^{k,i}$ the rank of $^{\ell}U_I^{k,i}$ within the set $\left\lbrace {^{\ell}U}_I^{j,i}\right\rbrace _{j=1}^n$. Let for $I \in \tau \setminus \varnothing$ the permutations $^{\ell}p_I$, $\ell = 1,\ldots,n$, be defined through $^{\ell}p_{I,i}(k)={^{\ell}R}_I^{k,i}$, $k=1,\ldots,n$ and denote by  $^{\ell}p_{I,i}^{-1}(\cdot)$ their inverses. 
\begin{enumerate}[resume] 
\item Recursively define for $\ell=1,\ldots,n$ samples $_{\ell}X_{I}^k$, $k=1,\ldots,n$, $I \in \mathscr{B}$ by
\begin{align}
\label{eq:scary1}
_{\ell}X_I^k = \sum_{J \in \mathscr{C}(I)} {^{\ell}X}_J^{^{\ell}q_J^{-1}(^{\ell}p_{J}(^{\ell}p_{I,1}^{-1}(^{\ell}q_{I,1}(k))))} = {^{\ell}X}_{I,1}^k+ \ldots + {^{\ell}X}_{I,N_I}^{^{\ell}q_{I,N_I}^{-1}(^{\ell}p_{I,N_I}(^{\ell}p_{I,1}^{-1}(^{\ell}q_{I,1}(k))))},
\end{align}
and samples $_{\ell}\boldsymbol{X}_{I}^k$ by
\begin{align}
\label{eq:scary2}
_{\ell}\boldsymbol{X}_I^k = \left( ^{\ell}\boldsymbol{X}_{I,1}^k,{^{\ell}\boldsymbol{X}}_{I,2}^{^{\ell}q_{I,2}^{-1}(^{\ell}p_{I,2}(^{\ell}p_{I,1}^{-1}(^{\ell}q_{I,1}(k))))},\ldots,{^{\ell}\boldsymbol{X}}_{I,N_I}^{^{\ell}q_{I,N_I}^{-1}(^{\ell}p_{I,N_I}(^{\ell}p_{I,1}^{-1}(^{\ell}q_{I,1}(k))))}\right),
\end{align}
where the permutations $^{\ell}q_J$ are defined through $^{\ell}q_J(k) = {^{\ell}Q}_J^k$ and  $^{\ell}Q_J^k$ denotes the rank of $^{\ell}X_J^k$ within the set $\left\lbrace ^{\ell}X_J^j\right\rbrace _{j=1}^n$.
\item For $k=1,\ldots,n$ set
\begin{align}
\label{eq:algArbenz1}
&X_I^k:={_{k}X}_I^k, \\
\label{eq:algArbenz2}
&\boldsymbol{X}_I^k := {_{k}}\boldsymbol{X}_I^k.
\end{align}
\item If we haven not yet reached the root $\varnothing$, repeat the algorithm up to this point $n$ different times to generate $n$ independent copies
\begin{itemize}
\item $\left\lbrace {^{1}X}_I^k \right\rbrace_{k=1}^n,\ldots,\left\lbrace {^{n}X}_I^k \right\rbrace_{k=1}^n$ of the reordered samples $\left\lbrace X_I^k \right\rbrace_{k=1}^n$;
\item $\left\lbrace ^{1}\boldsymbol{X}_I^k \right\rbrace_{k=1}^n,\ldots,\left\lbrace^{n}\boldsymbol{X}_I^k \right\rbrace_{k=1}^n$ of the reordered samples $\left\lbrace \boldsymbol{X}_I^k \right\rbrace_{k=1}^n$.
\end{itemize}
\item Repeat steps 2-4 till you reach the root node $\varnothing$ of the tree.
\end{enumerate}
\end{algorithm}
\begin{remark} \hspace{80 pt}
\label{re:MRA}
\begin{enumerate}
\item In step 2 of the MRA, for fixed $\ell$, we simply define recursively the reordered samples. The reordering order is chosen in such a way that the first component of the reordered samples is fixed (compare "Reordering 2" in Example \ref{ex:altReordering})
\item Since $\ell$ is running from $\ell=1,\ldots,n$, we obtain $n$ independent sets of reordered samples. In step 3 we use a diagonalizing procedure similar to (\ref{eq:diagonalize}), which yields $n$ i.i.d. samples.
\item In case we have not yet reached the root node, we need $n$ independent copies of these i.i.d. samples, so that we can continue with the next recursion step. These independent copies are generated in step 4.
\end{enumerate}
\end{remark}
\section{Convergence results}
In this section we discuss the convergence properties of the MRA. First, we show in Subsection \ref{sec:convergence1} that the results in Theorem \ref{the:convergencearbenz} are also true for the MRA. In addition to this, the MRA allows us to prove a result that goes beyond the ones already existing. We state and prove this convergence result in Subsection \ref{sec:convergence2}.
\subsection{The basic convergence result}
\label{sec:convergence1}
Due to the modifications carried out it is not a priori clear that the convergence result in Theorem \ref{the:convergencearbenz} are still true for the MRA. We mentioned in the previous section that the first modification - namely the change of reordering order - does not affect this result, since we obtain the same reordered samples (possibly in a different order). \\
The second and more significant modification may, however, affect the result. Fortunately, we can show that this is not the case when we additionally assume that the marginals are discrete (finite or infinite). This will be the subject of this section.
\\ \\
Let in the following $G_I^n$ and $F_I^n$, $\ell=1,\ldots,n$, for $I \in \mathscr{B}$ denote the empirical distribution functions defined through
\begin{align}\label{eq:GempMRA}
&G_I^n(x_1,\ldots,x_{N_I})= \frac{1}{n} \sum\limits_{k=1}^{n}\mathbb{1}_{\left\lbrace  {^{k}X}_{I,i}^{^{k}q_{I,i}^{-1}(^{k}p_{{I,i}}(^{k}p_{I,1}^{-1}({^{k}q}_{I,1}(k))))} \leq x_i \text{ for all } i=1,\dots,N_I \right\rbrace }, \\
\label{eq:FempMRA}
&F_I^n(x) = \frac{1}{n} \sum\limits_{k=1}^{n}\mathbb{1}_{\left\lbrace {X}_I^k \leq x\right\rbrace }.
\end{align}
Since especially the first expression looks confusing, we want to emphasis that these two expressions are nothing else than empirical distribution functions of the MRA-reordered samples (analogously to the expressions (\ref{eq:theorem_convergenceArbenz1})\&(\ref{eq:theorem_convergenceArbenz2}) in Theorem \ref{the:convergencearbenz})

\begin{theorem}
\label{the:convergence1}
Assume that the conditions on the copulas $C_I$, $I \in \mathscr{B}$, as formulated in Theorem \ref{the:convergencearbenz} are satisfied and that all the marginals are discrete. Then, for each branching node $I \in \mathscr{B}$
\begin{align}
\label{eq:convergence1.1}
&\lim\limits_{n \rightarrow \infty} \sup\limits_{\boldsymbol{x} \in \mathbb{R}^{N_I}} \left| G_I^n(\boldsymbol{x}) - G_I(\boldsymbol{x})\right|  = 0 \quad P-a.s., \\
\label{eq:convergence1.2}
&\lim\limits_{n \rightarrow \infty} \sup\limits_{t \in \mathbb{R}} \left| F_I^n(t) - F_I(t)\right|  = 0 \quad P-a.s.
\end{align}
\end{theorem}
The following auxiliary results will help us to prove Theorem \ref{the:convergence1}.

\paragraph{Auxiliary result 1}
\label{par:conv1_aux1}
Let $X_n: \Omega \rightarrow \mathbb{R}$, $n \in \mathbb{N}$, and $X: \Omega \rightarrow \mathbb{R}$ be \textit{discrete} random variables such that $X_n \xrightarrow{n \to \infty} X$ in distribution. Denote by ran$X=\left\lbrace X(\omega) | \omega \in \Omega \right\rbrace$ the range of $X$ and assume there exists a discrete set $E \subset \mathbb{R}$ such that ran$X \subset E$ and ran$X_n \subset E$, $\forall n \in \mathbb{N}$. \\
Then it holds that $\lim\limits_{n \to \infty}P[X_n=x]=P[X=x]$, $\forall x \in \mathbb{R}$.
\begin{proof}
Let first $x \in$ ran$X$. For a left-open, right-closed interval $U=(a,b]$ it holds that
\begin{align*}
\lim\limits_{n \to \infty}P[X_n\in U]=P[X\in U]
\end{align*}
by convergence in distribution. Since $E$ is discrete we can find $U=(a,b]$ such that $a<x<b$ and $E \cap U = \left\lbrace x\right\rbrace $. Then it holds that
\begin{align*}
&P[X_n \in U]=P[X_n=x], \quad P[X \in U]=P[X=x].
\end{align*}
If $x \notin$ ran$X$ we argue similarly. Find $U=(a,b]$ such that $a<x<b$ and $E \cap U = \varnothing$. Then $0=\lim\limits_{n \to \infty}P[X_n=x]=\lim\limits_{n \to \infty}P[X_n\in U]=P[X\in U]=P[X=x]=0$.
\end{proof}

\paragraph{Auxiliary result 2}
\label{par:conv1_aux2}
Let $X$ and $X_n$, $n\in \mathbb{N}$, as in the 1\textsuperscript{st} auxiliary result above. Denote their distribution functions by $F$, respectively $F_n$, and the associated measures by $\mu$, respectively $\mu_n$. Then the convergence in distribution is uniformly in the sense that
\begin{align}
\label{eq:conv1_uniform}
\lim\limits_{n \to \infty} \sup\limits_{x \in \mathbb{R}} \left| F_n(x) - F(x)\right|  = 0.
\end{align}
\begin{proof}
Let $\epsilon > 0$. Since $X$ is discrete there exists a finite set $A\subset E$, of cardinality $\left|A \right| < \infty$,  such that $P[X \in A]=\mu(A) > 1- \epsilon$. Choose $N_1 \in \mathbb{N}$ large enough such that $\forall n \geq N_1$: $\left| \mu_n(\mathbb{R}\setminus A) - \mu(\mathbb{R}\setminus A)\right| < \epsilon$. Then the triangle-inequality implies that $\forall n \geq N_1$:
\begin{align*}
\left| \mu_n(\mathbb{R}\setminus A) \right| < \epsilon + \left| \mu(\mathbb{R}\setminus A)\right| < 2\epsilon.
\end{align*}
By the 1\textsuperscript{st} auxiliary result above we know that $\lim\limits_{n \to \infty}\mu_n(x)=\lim\limits_{n \to \infty}P[X_n=x]=P[X=x]=\mu(x)$, for all $x \in A$. This, and the fact that $A$ is finite implies that there exists $N_2 \in \mathbb{N}$ such that $\forall n \geq N_2$: $\left| \mu_n(x)  - \mu(x) \right| < \epsilon$, $\forall x \in A$. 
\\ \\
Let now $x \in \mathbb{R}$ be arbitrary and set $B:=(-\infty,x] \cap A$ and $B^c:=(-\infty,x] \setminus A$. Then $\forall n \geq \max \left\lbrace N_1,N_2\right\rbrace$:
\begin{align*}
&\left|F_n(x)-F(x)\right| = \left| \mu_n(B)+\mu_n(B^c) - \mu(B)-\mu(B^c)\right| \\
&\leq \left| \mu_n(B) -\mu(B)\right| + \left|  \mu_n(B^c) -\mu(B^c)\right| \leq \left|A \right|\epsilon + 2\epsilon.
\end{align*}
\end{proof}

\paragraph{Auxiliary result 3}
\label{par:conv1_aux3}
Let $X$ and $X_n$, $n\in \mathbb{N}$, as in the 1\textsuperscript{st} auxiliary result above. Denote their distribution functions by $F$, respectively $F_n$. Let $F_{m,n}(x):=\sum\limits_{i=1}^{m} \mathbb{1}_{\left\lbrace X_n^i \leq x \right\rbrace }$, $(m,n) \in \mathbb{N}^2$, denote the empirical distribution function of the i.d.d. random variables $X_n^1,\ldots,X_n^m$ $\sim F_n$. Then
\begin{align}
\lim\limits_{n \to \infty} \sup\limits_{x \in \mathbb{R}} \left| F_{n,n}(x) - F(x)\right|  = 0 \quad P-a.s..
\end{align}
\begin{proof}
We write $\lVert F_{m,n} - F \rVert_{\infty} :=  \sup\limits_{x \in \mathbb{R}} \left| F_{m,n}(x) - F(x)\right|$ for ease of notation. Note that the Glivenko–Cantelli Theorem implies that $\lim\limits_{m \to \infty} \lVert F_{m,n} - F_n \rVert_{\infty}=0$ P-a.s., for all $n \in \mathbb{N}$. By the 2\textsuperscript{nd} auxiliary result above we know that $\forall k \in \mathbb{N}$ there exists $N(k) \in \mathbb{N}$ such that
\begin{align*}
\lVert F_n - F \rVert_{\infty} \leq \frac{1}{k}, \quad \forall n \geq N(k).
\end{align*}
Using these two facts, we get that for all $n \geq N(k)$:
\begin{align*}
\lim\limits_{m \to \infty} \lVert F_{m,n} - F \rVert_{\infty} \leq \underbrace{\lim\limits_{m \to \infty} \lVert F_{m,n} - F_n \rVert_{\infty}}_{\text{$\rightarrow 0$ P-a.s.}} + \underbrace{\lim\limits_{m \to \infty} \lVert F_{n} - F \rVert_{\infty}}_{\text{$\leq\frac{1}{k}$}} \leq \frac{1}{k} \quad P-a.s..
\end{align*}
So $\lim\limits_{\substack{n \to \infty \\ n \geq N(k)}} \lVert F_{n,n} - F \rVert_{\infty} \leq \frac{1}{k}$ P-a.s. Letting $k \to \infty$ then proves the assertion.
\end{proof}

\begin{proof}[Theorem \ref{the:convergence1}]
The proof of Theorem \ref{the:convergence1} is obtained by induction from the bottom to the top of the tree. The $F_I^n$ for leaf nodes $I \in \mathscr{L}$ are obtained through a simulation from the true distributions $F_I$. Therefore, the Glivenko–Cantelli Theorem directly yields (\ref{eq:convergence1.2}) for $I \in \mathscr{L}$.
\\ \\
Fix $I \in \mathscr{B}$. Suppose that (\ref{eq:convergence1.1}) and (\ref{eq:convergence1.2}) hold true for $(I,1),\ldots,(I,N_I)$. Recall the recursion procedure in step 2 of the MRA together with the second remark in (\ref{re:MRA}) and note that the MRA yields $n$ independent sets of random vectors 
\begin{align}
\begin{split}
\label{eq:copies}
&^{1}\boldsymbol{Y}_I^{(1:n)},\ldots,^{1}\boldsymbol{Y}_I^{(n:n)} \sim G_I^{*n} \qquad \text{ not i.i.d},\\
&\hspace{60 pt} \vdots \\
&^{n}\boldsymbol{Y}_I^{(1:n)},\ldots,^{n}\boldsymbol{Y}_I^{(n:n)} \sim G_I^{*n} \qquad \text{ not i.i.d}.
\end{split}
\end{align}
\subparagraph{Note} In practice, the MRA yields realizations of these random variables. Concretely, realizations of $^{\ell}\boldsymbol{Y}_I^{(k:n)}$, $(\ell,k) \in \left\lbrace 1,\ldots,n\right\rbrace ^2$, are given by
\begin{align*}
^{\ell}\boldsymbol{Y}_I^{(k:n)}:=\left( ^{\ell}{X}_{I,1}^k,{^{\ell}{X}}_{I,2}^{^{\ell}q_{I,2}^{-1}(^{\ell}p_{I,2}(^{\ell}p_{I,1}^{-1}(^{\ell}q_{I,1}(k))))},\ldots,{^{\ell}{X}}_{I,N_I}^{^{\ell}q_{I,N_I}^{-1}(^{\ell}p_{I,N_I}(^{\ell}p_{I,1}^{-1}(^{\ell}q_{I,1}(k))))}\right).
\end{align*}
To avoid that the complicated form distracts us from the idea of the proof, we decided to mention this just as a side note.
\\ \\
We stress that all $^{\ell}\boldsymbol{Y}_I^{(k:n)} \in \mathbb{R}^{N_I}$, $(\ell,k) \in \left\lbrace 1,\ldots,n \right\rbrace^2 $, indeed have the same distribution, which we denote by $G_I^{*n}$. The $*n$ is used to distinguish the distribution function $G_I^{*n}$ from the \textit{empirical} distribution function $G_I^{n}$.
\\ \\
Denote by $^{\ell}G_I^n$, $\ell=1,\ldots,n$, the empirical distribution function of the identically but not independently distributed random vectors $^{\ell}\boldsymbol{Y}_I^{(1:n)},\ldots,{^{\ell}\boldsymbol{Y}}_I^{(n:n)}$. These random vectors $^{\ell}\boldsymbol{Y}_I^{(1:n)},\ldots,{^{\ell}\boldsymbol{Y}}_I^{(n:n)}$ were reordered according to Arbenz' reordering algorithm (except for the different reordering order). Hence, we can apply Lemma 3.6 in \cite{arbenz12} in combination with the induction hypothesis, which tells us that for all $\ell=1,\ldots,n$ it holds that
\begin{align}
\label{eq:conv1arbenz}
\lim\limits_{n \rightarrow \infty} \sup\limits_{\boldsymbol{x} \in \mathbb{R}^{N_I}} \left| {^{\ell}G}_I^n(\boldsymbol{x}) - G_I(\boldsymbol{x})\right|  = 0 \quad P-a.s.
\end{align}
Recall step 3 in the MRA and the related second remark in (\ref{re:MRA}): The MRA picks the the diagonal elements ${^{1}\boldsymbol{Y}}_I^{(1:n)},\ldots,{^{k}\boldsymbol{Y}}_I^{(k:n)},\ldots,{^{n}\boldsymbol{Y}}_I^{(n:n)}$ in (\ref{eq:copies}). We define
\begin{align*}
&\boldsymbol{Y}_I^{(1:n)}:={^{1}\boldsymbol{Y}}_I^{(1:n)}, \\
&\hspace{35 pt} \vdots \\
&\boldsymbol{Y}_I^{(n:n)}:={^{n}\boldsymbol{Y}}_I^{(n:n)}.
\end{align*}
Hence, the MRA yields i.i.d. random vectors $\boldsymbol{Y}_I^{(1:n)},\ldots,\boldsymbol{Y}_I^{(n:n)} \sim G_I^{*n}$. Note that the empirical distribution function of the random vectors $\boldsymbol{Y}_I^{(1:n)},\ldots,\boldsymbol{Y}_I^{(n:n)}$ is $G_I^n$ from (\ref{eq:GempMRA}). Hence, it remains to show that $G_I^n$ converges to $G_I$ in the sense of (\ref{eq:convergence1.1}).
\\ \\
We prove this by showing that $\boldsymbol{Y}_I^{(k:n)} \sim G_I^{*n}$, $k=1,\ldots,n$, converges weakly towards $G_I$, i.e. we need to show that
\begin{align}
\label{eq:conv1_weakconv}
\lim\limits_{n \rightarrow \infty}G_I^{*n}(\boldsymbol{x})=G_I(\boldsymbol{x}), \qquad \forall \boldsymbol{x}\in \mathbb{R}^{N_I} \text{ at which } G_I \text{ is continuous.}
\end{align}
If (\ref{eq:conv1_weakconv}) holds true, we can apply the 3\textsuperscript{rd} auxiliary result \ref{par:conv1_aux3} which tells us that the empirical distribution function $G_I^n$ of the i.i.d. random vectors $\boldsymbol{Y}_I^{(1:n)},\ldots,\boldsymbol{Y}_I^{(n:n)} \sim G_I^{*n}$ converges towards $G_I$ in the sense of (\ref{eq:convergence1.1}).
\\
Moreover, note that $X_I^k=$"sum of the components of $\boldsymbol{Y}_I^{(k:n)}$", $k=1,\ldots,n$,  and $F_I^n$ is the empirical distribution function of the random variables $X_I^1,\ldots,X_I^n$ (recall \ref{eq:FempMRA}). We can therefore apply the Cramér–Wold theorem and then argue again as above with the 3\textsuperscript{rd} auxiliary result to prove (\ref{eq:convergence1.2}). 
\\ \\
It remains to show that (\ref{eq:conv1_weakconv}) is true, which we prove by contradiction: Assume $\exists \boldsymbol{x}\in\mathbb{R}^{N_I}$ such that $\exists\varepsilon>0$ s.t. $\forall n \in \mathbb{N}$ $\exists N(n) > n$: $|G_I^{*N(n)}(\boldsymbol{x}) -  G_I(\boldsymbol{x})|>\varepsilon$. Fix an arbitrary $ \ell \in \left\lbrace1,\ldots,n \right\rbrace$. Then
\begin{align*}
E[|G_I(\boldsymbol{x})-{^{\ell}G}_I^{N(n)}(\boldsymbol{x})|] &\geq |E[G_I(\boldsymbol{x})-{^{\ell}G}_I^{N(n)}(\boldsymbol{x})]| \\
&\geq |E[G_I(\boldsymbol{x})]-E[{^{\ell}G}_I^{N(n)}(\boldsymbol{x})]| \\
&\geq |G_I(\boldsymbol{x}) - E[\frac{1}{N(n)}\sum\limits_{k=1}^{N(n)} \mathbb{1}_{\left\lbrace \boldsymbol{^{\ell}{Y}}_I^{(k:N(n))}\leq \boldsymbol{x}\right\rbrace }]| \\
&=|G_I(\boldsymbol{x}) - G_I^{*N(n)}(\boldsymbol{x})|>\varepsilon
\end{align*}
This is a contradiction because ${^{\ell}G}_I^{N(n)}(\boldsymbol{x})\xrightarrow{n \to \infty} G_I(\boldsymbol{x})$ P-a.s. by (\ref{eq:conv1arbenz}) and so by dominated convergence $E[|G_I(\boldsymbol{x})-{^{\ell}G}_I^{N(n)}(\boldsymbol{x})|]\xrightarrow{n \to \infty} 0$.
\end{proof}
\subsection{Convergence towards the tree dependent distribution}
\label{sec:convergence2}
In the previous section we showed that in case of discrete marginals the MRA at least has the same convergence properties as the original algorithm. We will now go a step further and claim that in case of discrete marginals the MRA yields approximations of the unique tree dependent random vector.
\begin{figure}[h]
\centering
\includegraphics[width=1\linewidth]{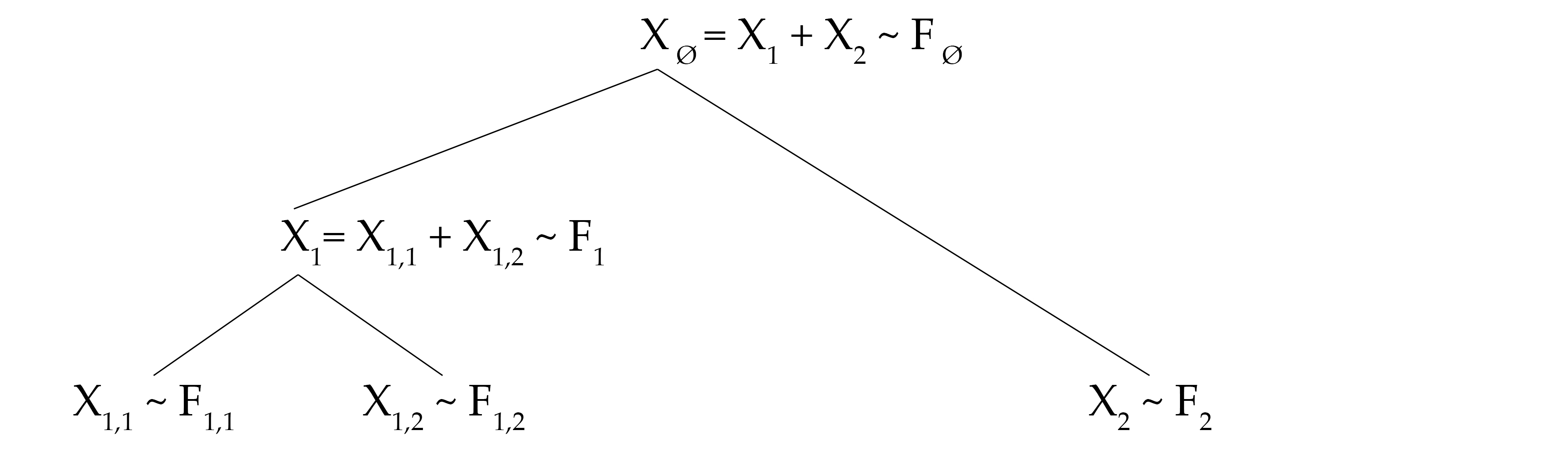}
\caption{An illustration of the aggregation tree model (\ref{eq:con2_model}).}
\label{fig:3dimtree1}
\end{figure}

In order to keep the proof of this claim as comprehensible as possible we decided to restrict ourselves to one of the simplest tree structures and refer to the Appendix \ref{appendix} where we discuss an extension to more general aggregation trees. More precisely, we will consider the aggregation tree model
\begin{align} \label{eq:con2_model}
\left( \tau,(F_I)_{I\in \mathscr{L}(\tau)},(C)_{I \in \mathscr{B}(\tau)}\right),
\end{align}
with $\tau = \left\lbrace \varnothing, ~(1), ~(1,1), ~(1,2), ~(2) \right\rbrace$; three univariate distribution functions $F_{1,1}$, $F_{1,2}$, $F_{2}$ of \textit{discrete} random variables ; and two bivariate copulas $C_1$ and $C_{\varnothing}$ satisfying the assumptions of Theorem \ref{the:convergencearbenz}. Figure \ref{fig:3dimtree1} illustrates the aggregation tree model (\ref{eq:con2_model}) for which we will prove the following theorem:
\begin{theorem}
\label{the:conv2}
Consider the aggregation tree model defined in (\ref{eq:con2_model}). Let $\boldsymbol{X}_{\varnothing}=(X_{1,1},X_{1,2},X_2)$ be the unique tree dependant random vector associated with this model and denote the distribution function of $\boldsymbol{X}_{\varnothing}$ by $F_{\tau}$. \\
Then, the MRA yields realizations of i.i.d. reordered random vectors which all satisfy the conditional independence assumption (\ref{def:CIA}) and additionally converge in distribution towards the unique tree dependent random vector $\boldsymbol{X}_{\varnothing}$. In particular, if $F^n_{\tau}$ denotes the empirical distribution functions of these reordered random vectors, then it holds that
\begin{align}
\label{eq:theorem_conv2}
\lim\limits_{n \rightarrow \infty} \sup\limits_{\boldsymbol{x} \in \mathbb{R}^3} \left| F_{\tau}^n(\boldsymbol{x}) - F_{\tau}(\boldsymbol{x})\right|  = 0 \quad P-a.s.
\end{align}
\end{theorem}
\subsubsection{Preliminaries}
\label{par:preliminaries}
From now on we will think of all the samples appearing in the MRA \ref{the:MRA} not as samples anymore but as random variables. It should be clear that in practice the MRA then yields realizations (samples) of these random variables. Our aim will be to gain a better understanding of the expressions (\ref{eq:scary1}) and (\ref{eq:scary2}) in the 2\textsuperscript{nd} step of the MRA.
\\ \\
We stress first that the ranks of a set of i.i.d. random variables are as well random variables. Let for instance $Z^1,\ldots,Z^n$ be a sequence of i.i.d. random variables. Recall that the rank $R^k$ of $Z^k$ within the set $\left\lbrace Z^j \right\rbrace _{j=1}^n$ is defined as:
\begin{align}
\label{eq:conv2_rank}
R_k := \sum\limits_{j=1}^{n}\mathbb{1}_{\left\lbrace Z^j \leq Z^k \right\rbrace }.
\end{align}
Expression (\ref{eq:conv2_rank}) is clearly a random variable. Thus, (\ref{eq:scary1}) and (\ref{eq:scary2}) must as well be random variables.
\\ \\
In the present compact form used in the MRA it seems difficult to talk about properties of the expression (\ref{eq:scary2}). We therefore propose a different representation of (\ref{eq:scary2}). The following short example illustrates the idea.
\begin{example}
Let $Z^1, Z^2$ i.i.d. and denote by $R^k$ the rank of $Z^k$. Define the permutation $p(k)=R^k$, $k=1,2$. Let further $V^1, V^2$ i.i.d. be some random variables. Assume the random variables $X^1, X^2$ are defined as follows:
\begin{align} \label{eq:conv2_compact}
X^1:=V^{p(1)}, \hspace{30 pt} X^2:=V^{p(2)}.
\end{align}
Observe that we could alternatively write
\begin{align}
\begin{split} \label{eq:conv2_noncompact}
X^1 = V^1 \mathbb{1}_{\left\lbrace Z^1 < Z^2\right\rbrace } + V^2\mathbb{1}_{\left\lbrace Z^1 \geq Z^2\right\rbrace },\\
X^2 = V^2 \mathbb{1}_{\left\lbrace Z^1 \leq Z^2\right\rbrace } + V^1\mathbb{1}_{\left\lbrace Z^1 > Z^2\right\rbrace }.
\end{split}
\end{align}
\end{example}
Obviously, (\ref{eq:conv2_compact}) is a compact version of (\ref{eq:conv2_noncompact}). Representation (\ref{eq:conv2_noncompact}) has the advantage that it seems easier to analyse its distributional properties. Analogously to (\ref{eq:conv2_noncompact}), we will now give an extended representation of
\begin{align*}
\left( ^{1}\boldsymbol{X}_{1}^1,{^{1}\boldsymbol{X}}_{2}^{^{1}q_{\varnothing,2}^{-1}(^1p_{\varnothing,2}(^1p_{\varnothing,1}^{-1}(^{1}q_{\varnothing,1}(1))))}\right)
\end{align*}
for $n=2$. This is expression (\ref{eq:scary2}) in our aggregation tree model (\ref{eq:con2_model}) with $\ell=1$, $k=1$ and $n=2$. We focus on the second component only:
\begin{align*}
 {^{1}\boldsymbol{X}}_{2}^{^{1}q_{\varnothing,2}^{-1}(^1p_{\varnothing,2}(^1p_{\varnothing,1}^{-1}(^{1}q_{\varnothing,1}(1))))} &= {^1X}_2^1 \mathbb{1}_{\left\lbrace {^1X}_1^1 \geq {^1X}_1^2\right\rbrace \cap \left\lbrace {^1U}_{\varnothing}^{1,1} \geq {^1U}_{\varnothing}^{2,1}\right\rbrace \cap \left\lbrace {^1U}_{\varnothing}^{1,2} \geq {^1U}_{\varnothing}^{2,2}\right\rbrace \cap \left\lbrace {^1X}_2^1 \geq {^1X}_2^2\right\rbrace} \\
 &+{^1X}_2^2\mathbb{1}_{\left\lbrace {^1X}_1^1 \geq {^1X}_1^2\right\rbrace \cap \left\lbrace {^1U}_{\varnothing}^{1,1} \geq {^1U}_{\varnothing}^{2,1}\right\rbrace \cap \left\lbrace {^1U}_{\varnothing}^{1,2} \geq {^1U}_{\varnothing}^{2,2}\right\rbrace \cap \left\lbrace {^1X}_2^1 < {^1X}_2^2\right\rbrace} \\
 &\hspace{142 pt} \vdots \\
 &+{^1X}_2^1\mathbb{1}_{\left\lbrace {^1X}_1^1 < {^1X}_1^2\right\rbrace \cap \left\lbrace {^1U}_{\varnothing}^{1,1} < {^1U}_{\varnothing}^{2,1}\right\rbrace \cap \left\lbrace {^1U}_{\varnothing}^{1,2} < {^1U}_{\varnothing}^{2,2}\right\rbrace \cap \left\lbrace {^1X}_2^1 < {^1X}_2^2\right\rbrace}.
\end{align*}
Note that in the above sum appear $2^4=16$ different indicator functions, and so the sum consists of 16 summands. In case of $n=3$, the sum would already comprise $(3!)^4=1296$ different summands.\\ 
For the proof of Theorem \ref{the:conv2} it is important to note that - no matter how large we choose $n$ - we can always write the above sum as a random function $f({^1X}_1^1)$ in ${^1X}_1^1$. For $n=2$ we have, for instance,
\begin{align*}
 f(x) &= {^1X}_2^1 \mathbb{1}_{\left\lbrace x \geq {^1X}_1^2\right\rbrace \cap \left\lbrace {^1U}_{\varnothing}^{1,1} \geq {^1U}_{\varnothing}^{2,1}\right\rbrace \cap \left\lbrace {^1U}_{\varnothing}^{1,2} \geq {^1U}_{\varnothing}^{2,2}\right\rbrace \cap \left\lbrace {^1X}_2^1 \geq {^1X}_2^2\right\rbrace} \\
 &+{^1X}_2^2\mathbb{1}_{\left\lbrace x \geq {^1X}_1^2\right\rbrace \cap \left\lbrace {^1U}_{\varnothing}^{1,1} \geq {^1U}_{\varnothing}^{2,1}\right\rbrace \cap \left\lbrace {^1U}_{\varnothing}^{1,2} \geq {^1U}_{\varnothing}^{2,2}\right\rbrace \cap \left\lbrace {^1X}_2^1 < {^1X}_2^2\right\rbrace} \\
 &\hspace{142 pt} \vdots \\
 &+{^1X}_2^1\mathbb{1}_{\left\lbrace x < {^1X}_1^2\right\rbrace \cap \left\lbrace {^1U}_{\varnothing}^{1,1} < {^1U}_{\varnothing}^{2,1}\right\rbrace \cap \left\lbrace {^1U}_{\varnothing}^{1,2} < {^1U}_{\varnothing}^{2,2}\right\rbrace \cap \left\lbrace {^1X}_2^1 < {^1X}_2^2\right\rbrace}.
\end{align*}
The random function $f(x)$ consists of the random variables ${^1X}_2^1$, ${^1X}_2^2$, ${^1X}_1^2$, ${^1U}_{\varnothing}^{1,1}$, ${^1U}_{\varnothing}^{2,1}$, ${^1U}_{\varnothing}^{1,2}$, ${^1U}_{\varnothing}^{2,2}$, ${^1X}_2^1$, ${^1X}_2^2$ which are all \textit{independent} of ${^1X}_1^1$, and hence $f(x)$ is \textit{independent} of ${^1X}_1^1$.

\paragraph{Auxiliary result 4}
\label{par:conv2_aux4}
Based on the 1\textsuperscript{st} auxiliary result \ref{par:conv1_aux1} we state the following extension: Let the discrete random vectors $(X_n,Y_n)\xrightarrow{n \to \infty}(X,Y)$ converge in distribution and assume again that there exists a discrete set $E \subset \mathbb{R}^2$ which contains the ranges ran$(X_n,Y_n)$ and ran$(X,Y)$ . Fix $(x,y) \in \mathbb{R}\times \text{ran}Y$, then
\begin{align*}
\lim\limits_{n \to \infty}E[\mathbb{1}_{\left\lbrace X_n \leq x \right\rbrace } | Y_n=y ]=E[\mathbb{1}_{\left\lbrace X \leq x \right\rbrace } | Y=y ].
\end{align*}
Note that for the proof we can proceed analogously as in the proof of the 1\textsuperscript{st} auxiliary result \ref{par:conv1_aux1}. We do not want to dwell on this any further and instead proceed with the proof of our main Theorem \ref{the:conv2}.
\begin{proof}[Theorem \ref{the:conv2}]

Due to the notational complexity we strive for a clear and logical organization. We strongly suggest to grasp the "big picture" first before getting lost in technical details.\\
Denote by $\boldsymbol{X}_{\varnothing}^{(1:n)}=\left( X_{1,1}^{(1:n)},X_{1,2}^{(1:n)},X_{2}^{(1:n)}\right) ,\ldots,\boldsymbol{X}_{\varnothing}^{(n:n)}=\left( X_{1,1}^{(n:n)},X_{1,2}^{(n:n)},X_{2}^{(n:n)}\right)$ the reordered random vectors (\ref{eq:algArbenz2}) obtained in step 3 of the MRA. More precisely, $\boldsymbol{X}_{\varnothing}^{(1:n)},\ldots, \boldsymbol{X}_{\varnothing}^{(n:n)}$ are given by
\begin{align*}
&\left( X_{1,1}^{(1:n)},X_{1,2}^{(1:n)},X_{2}^{(1:n)}\right):=\boldsymbol{X}_{\varnothing}^1=\left( ^{1}\boldsymbol{X}_{1}^1,{^{1}{X}}_{2}^{^{1}q_{\varnothing,2}^{-1}(^1p_{\varnothing,2}(^1p_{\varnothing,1}^{-1}(^{1}q_{\varnothing,1}(1))))}\right), \\
&\hspace{115 pt} \vdots\\
&\left( X_{1,1}^{(n:n)},X_{1,2}^{(n:n)},X_{2}^{(n:n)}\right):=\boldsymbol{X}_{\varnothing}^n=\left( ^{n}\boldsymbol{X}_{1}^n,{^{n}{X}}_{2}^{^{n}q_{\varnothing,2}^{-1}(^np_{\varnothing,2}(^np_{\varnothing,1}^{-1}(^{n}q_{\varnothing,1}(n))))}\right).
\end{align*}
Our overall goal is to show that $\boldsymbol{X}_{\varnothing}^{(1:n)}\xrightarrow{n \to \infty} \boldsymbol{X}_{\varnothing}$ in distribution. By the 3\textsuperscript{rd} auxiliary result \ref{par:conv1_aux3} we can then immediately conclude that the empirical distribution function $F_{\tau}^n$ of the i.i.d. random vectors $\boldsymbol{X}_{\varnothing}^{(1:n)},\ldots, \boldsymbol{X}_{\varnothing}^{(n:n)}$ converges towards $F_{\tau}$ in the sense of (\ref{eq:theorem_conv2}).
\paragraph{Step 1} We show that  $\boldsymbol{X}_{\varnothing}^{(1:n)},\ldots, \boldsymbol{X}_{\varnothing}^{(n:n)}$ satisfy the conditional independence assumption (\ref{def:CIA}). Since $\boldsymbol{X}_{\varnothing}^{(1:n)},\ldots, \boldsymbol{X}_{\varnothing}^{(n:n)}$ are i.i.d. it is enough to show that (\ref{def:CIA}) holds for $\boldsymbol{X}_{\varnothing}^{(1:n)}$. For this purpose we introduce the random variable $X_{1}^{(1:n)}=X_{1,1}^{(1:n)}+X_{1,2}^{(1:n)}$. Note that for $\boldsymbol{X}_{\varnothing}^{(1:n)}$ the conditional independence assumption then reads:
\begin{align}
\label{eq:conv2_condind}
(X_{1,1}^{(1:n)},X_{1,2}^{(1:n)},X_{1}^{(1:n)})\perp (X_{2}^{(1:n)},X_{1}^{(1:n)}+X_{2}^{(1:n)})\mid X_{1}^{(1:n)}.
\end{align}
The discussion in the Preliminaries \ref{par:preliminaries} is crucial for the proof of (\ref{eq:conv2_condind}). We argued that $X_{2}^{(1:n)}$ can be written as a function in $X_{2}^{(1:n)}$ by $X_{2}^{(1:n)} = f(X_{1}^{(1:n)})$, where $f(x)$ as a function in $x$ is a random function consisting of random variables that are all \textit{independent} of $X_{1}^{(1:n)}$. \\ \\
Let $g: \mathbb{R}^3\rightarrow \mathbb{R}$ and $h: \mathbb{R}^2\rightarrow \mathbb{R}$ be two measurable and bounded functions.
\begin{align*}
&E[g(X_{1,1}^{(1:n)},X_{1,2}^{(1:n)},X_{1}^{(1:n)})h(X_{2}^{(1:n)},X_{1}^{(1:n)}+X_{2}^{(1:n)})\mid X_{1}^{(1:n)}]\\
&= E[g(X_{1,1}^{(1:n)},X_{1,2}^{(1:n)},X_{1}^{(1:n)})h(f(X_{1}^{(1:n)}),X_{1}^{(1:n)}+f(X_{1}^{(1:n)}))\mid X_{1}^{(1:n)}] \\
&= E[g(X_{1,1}^{(1:n)},X_{1,2}^{(1:n)},x)h(f(x),x+f(x))\mid X_{1}^{(1:n)}]\mid_{x=X_{1}^{(1:n)}}  
\end{align*}
Using that $h(f(x),x+f(x))$ is a function of random variables that are all independent of $X_{1,1}^{(1:n)},X_{1,2}^{(1:n)}$ and $X_{1}^{(1:n)}$, we proceed by
\begin{align*}
&= E[h(f(x),x+f(x))]\mid_{x=X_{1}^{(1:n)}}E[g(X_{1,1}^{(1:n)},X_{1,2}^{(1:n)},x)\mid X_{1}^{(1:n)}]\mid_{x=X_{1}^{(1:n)}} \\
&= E[h(f(X_{1}^{(1:n)}),X_{1}^{(1:n)}+f(X_{1}^{(1:n)}))\mid X_{1}^{(1:n)}]E[g(X_{1,1}^{(1:n)},X_{1,2}^{(1:n)},X_{1}^{(1:n)})\mid X_{1}^{(1:n)}] \\
&=E[h(X_{2}^{(1:n)},X_{1}^{(1:n)}+X_{2}^{(1:n)})\mid X_{1}^{(1:n)}]E[g(X_{1,1}^{(1:n)},X_{1,2}^{(1:n)},X_{1}^{(1:n)})\mid X_{1}^{(1:n)}].
\end{align*} 
This proves (\ref{eq:conv2_condind}) and hence the first step.

\paragraph{Step 2} In the second and last step we prove that the reordered random vector $\boldsymbol{X}_{\varnothing}^{(1:n)}=\left( X_{1,1}^{(1:n)},X_{1,2}^{(1:n)},X_{2}^{(1:n)}\right)$ converges in distribution towards the unique tree dependent random vector $\boldsymbol{X}_{\varnothing}=(X_{1,1},X_{1,2},X_{2})$. \\ \\
Recall that we know from Theorem \ref{the:convergence1} that $\boldsymbol{X}_{\varnothing}^{(1:n)}$ approximates the distributions $G_I$ and $F_I$, $I \in \mathscr{B}$. More precisely, with $X_1 = X_{1,1}+X_{1,2}$ it holds that 
\begin{align}
\label{eq:conv2_arbenz1}
&(X_{1,1}^{(1:n)},X_{1,2}^{(1:n)},X_{1}^{(1:n)}) \xrightarrow{n \to \infty} (X_{1,1},X_{1,2},X_{1}), \\
\label{eq:conv2_arbenz2}
&(X_{2}^{(1:n)},X_{1}^{(1:n)}) \xrightarrow{n \to \infty} (X_{2},X_{1}) \quad \text{in distribution}.
\end{align}
Let $E=\left\lbrace \text{ran}X_{1,1},\text{ran}X_{1,2},\text{ran}X_{2} \right\rbrace \subset \mathbb{R}$ be the discrete set comprising the ranges of $X_{1,1}$, $X_{1,2}$ and $X_{2}$. Note that $\left\lbrace \text{ran}X_{1,1}^{(1:n)},\text{ran}X_{1,2}^{(1:n)},\text{ran}X_{2}^{(1:n)} \right\rbrace \subset E$, for all $n \in \mathbb{N}$. Then 
\begin{align*}
&X_{1} : (\varOmega, \mathscr{A}, \mathbb{P}) \rightarrow (E, \mathcal{P(E)}), \\
&X_{1}^{(1:n)} : (\varOmega, \mathscr{A}, \mathbb{P}) \rightarrow (E, \mathcal{P(E)}),
\end{align*}
where $\mathcal{P(E)}$ is the power set of $E$. Denote by $\mu=\mathbb{P}\circ X_1^{-1}$ and $\mu_{n}=\mathbb{P}\circ \left( {X_{1}^{(1:n)}}\right) ^{-1}$ the push-forward measures of $\mathbb{P}$ by $X_1$, respectively $X_{1}^{(1:n)}$. By (\ref{eq:conv2_arbenz1}), the sequence of measures $\mu_n$ on $(E, \mathcal{P(E)})$ converges weakly towards the measure $\mu$.
\\ \\
Keeping this in mind, we are now ready to prove the main claim. Let $(x_{1,1},x_{1,2},x_2)\in\mathbb{R}^3$. 
\begin{align*}
&P[X_{1,1}<x_{1,1},X_{1,2}<x_{1,2},X_{2}<x_{2}]=E[\mathbb{1}_{\left\lbrace X_{1,1}<x_{1,1}\right\rbrace}\mathbb{1}_{\left\lbrace X_{1,2}<x_{1,2}\right\rbrace}\mathbb{1}_{\left\lbrace X_{2}<x_{2}\right\rbrace}] \\
&=E\left[ E[\mathbb{1}_{\left\lbrace X_{1,1}<x_{1,1}\right\rbrace}\mathbb{1}_{\left\lbrace X_{1,2}<x_{1,2}\right\rbrace}\mathbb{1}_{\left\lbrace X_{2}<x_{2}\right\rbrace}\mid X_{1}]\right] 
\end{align*}
Using that $\boldsymbol{X}_{\varnothing}$ satisfies the conditional independence assumption (\ref{def:CIA}), we get
\begin{align*}
&=E\left[ E[\mathbb{1}_{\left\lbrace X_{1,1}<x_{1,1}\right\rbrace}\mathbb{1}_{\left\lbrace X_{1,2}<x_{1,2}\right\rbrace}\mid X_{1}] E[\mathbb{1}_{\left\lbrace X_{2}<x_{2}\right\rbrace}\mid X_{1}]\right] \\
&=\int_{E} \underbrace{E[\mathbb{1}_{\left\lbrace X_{1,1}<x_{1,1}\right\rbrace}\mathbb{1}_{\left\lbrace X_{1,2}<x_{1,2}\right\rbrace}\mid X_{1}=x] E[\mathbb{1}_{\left\lbrace X_{2}<x_{2}\right\rbrace}\mid X_{1}=x]}_{\text{{\normalsize  $=:h(x)$}}}d\mu(x) 
\end{align*}
Note that $h(x):=E[\mathbb{1}_{\left\lbrace X_{1,1}<x_{1,1}\right\rbrace}\mathbb{1}_{\left\lbrace X_{1,2}<x_{1,2}\right\rbrace}\mid X_{1}=x] E[\mathbb{1}_{\left\lbrace X_{2}<x_{2}\right\rbrace}\mid X_{1}=x]$ is a function in $x$ mapping from $E$ to $\mathbb{R}$. In fact, $h:E \rightarrow \mathbb{R}$ as such is \textit{continuous and bounded}. So the Portmanteau-Theorem implies
\begin{align*}
=\lim\limits_{m \to \infty}\int_{E} E[\mathbb{1}_{\left\lbrace X_{1,1}<x_{1,1}\right\rbrace}\mathbb{1}_{\left\lbrace X_{1,2}<x_{1,2}\right\rbrace}\mid X_{1}=x] E[\mathbb{1}_{\left\lbrace X_{2}<x_{2}\right\rbrace}\mid X_{1}=x]d\mu_m(x)
\end{align*}
Let  $h_n(x):=E[\mathbb{1}_{\left\lbrace X_{1,1}^{(1:n)}<x_{1,1}\right\rbrace}\mathbb{1}_{\left\lbrace X_{1,2}^{(1:n)}<x_{1,2}\right\rbrace}\mid X_{1}^{(1:n)}=x] E[\mathbb{1}_{\left\lbrace X_{2}^{(1:n)}<x_{2}\right\rbrace}\mid X_{1}^{(1:n)}=\nolinebreak x]$. The functions $h_n:E \rightarrow \mathbb{R}$  are \textit{continuous and bounded} and the 4\textsuperscript{th} auxiliary result \ref{par:conv2_aux4} implies that $\lim\limits_{n \to \infty}h_n(x)=h(x)$. Therefore, by dominated convergence 
\begin{align*}
&=\lim\limits_{m \to \infty}\lim\limits_{n \to \infty}\int_{E} \underbrace{E[\mathbb{1}_{\left\lbrace X_{1,1}^{(1:n)}<x_{1,1}\right\rbrace}\mathbb{1}_{\left\lbrace X_{1,2}^{(1:n)}<x_{1,2}\right\rbrace}\mid X_{1}^{(1:n)}=x] E[\mathbb{1}_{\left\lbrace X_{2}^{(1:n)}<x_{2}\right\rbrace}\mid X_{1}^{(1:n)}=x]}_{\text{{\normalsize  $=:h_n(x)$}}}d\mu_m(x)
\end{align*}
We claim (proof below) that the iterated limes $\lim\limits_{m \to \infty}\lim\limits_{n \to \infty}\int_{E} h_n(x)d\mu_m(x)$ equals the double limes $\lim\limits_{(m,n) \to \infty}\int_{E} h_n(x)d\mu_m(x)$ and hence equals $\lim\limits_{n \to \infty}\int_{E} h_n(x)d\mu_n(x)$.
\begin{align*}
&=\lim\limits_{n \to \infty}\int_{E} E[\mathbb{1}_{\left\lbrace X_{1,1}^{(1:n)}<x_{1,1}\right\rbrace}\mathbb{1}_{\left\lbrace X_{1,2}^{(1:n)}<x_{1,2}\right\rbrace}\mid X_{1}^{(1:n)}=x] E[\mathbb{1}_{\left\lbrace X_{2}^{(1:n)}<x_{2}\right\rbrace}\mid X_{1}^{(1:n)}=x]d\mu_n(x) \\
&=\lim\limits_{n \to \infty}E\left[ E[\mathbb{1}_{\left\lbrace X_{1,1}^{(1:n)}<x_{1,1}\right\rbrace}\mathbb{1}_{\left\lbrace X_{1,2}^{(1:n)}<x_{1,2}\right\rbrace}\mid X_{1}^{(1:n)}] E[\mathbb{1}_{\left\lbrace X_{2}^{(1:n)}<x_{2}\right\rbrace}\mid X_{1}^{(1:n)}]\right]
\end{align*}
Note that according to the first step, $\boldsymbol{X}_{\varnothing}^{(1:n)}=\left( X_{1,1}^{(1:n)},X_{1,2}^{(1:n)},X_{2}^{(1:n)}\right)$ satisfies the conditional independence assumption (\ref{eq:conv2_condind}) and so
\begin{align*}
&=\lim\limits_{n \to \infty}E\left[ E[\mathbb{1}_{\left\lbrace X_{1,1}^{(1:n)}<x_{1,1}\right\rbrace}\mathbb{1}_{\left\lbrace X_{1,2}^{(1:n)}<x_{1,2}\right\rbrace}\mathbb{1}_{\left\lbrace X_{2}^{(1:n)}<x_{2}\right\rbrace}\mid X_{1}^{(1:n)}]\right] \\
&=\lim\limits_{n \to \infty}E[\mathbb{1}_{\left\lbrace X_{1,1}^{(1:n)}<x_{1,1}\right\rbrace}\mathbb{1}_{\left\lbrace X_{1,2}^{(1:n)}<x_{1,2}\right\rbrace}\mathbb{1}_{\left\lbrace X_{2}^{(1:n)}<x_{2}\right\rbrace}] \\
&=\lim\limits_{n \to \infty}P[X_{1,1}^{(1:n)}<x_{1,1},X_{1,2}^{(1:n)}<x_{1,2},X_{2}^{(1:n)}<x_{2}].
\end{align*}
This proves the second step and hence the whole Theorem \ref{the:conv2}. It only remains to show that the above claim holds true. Let us formulate the claim once again for the sake of clarity.
\paragraph{Claim} Let $h$ and $h_n$, $n \in \mathbb{N}$, be continuous and bounded functions ($\lVert h \rVert_{\infty}$, $\lVert h_n \rVert_{\infty} \leq 1$) mapping from a discrete set $E$ to $\mathbb{R}$ and suppose that the sequence $h_n$ converges pointwise towards $h$. Let $F$ and $F_m$, $m \in \mathbb{N}$, be distribution functions of discrete random variables on $E$ and denote the associated measures on $E$ by $\mu$, respectively $\mu_m$. Suppose $F_m$ converges uniformly towards $F$ as in (\ref{eq:conv1_uniform}). Then the following equality holds:
\begin{align*}
\lim\limits_{m \to \infty}\lim\limits_{n \to \infty}\int_{E} h_n(x)d\mu_m(x) = \lim\limits_{(m,n) \to \infty}\int_{E} h_n(x)d\mu_m(x).
\end{align*}
\begin{proof}
Let $L_{m,n}:=\int_{E} h_n(x)d\mu_m(x)$, $L:=\int_{E} h(x)d\mu(x)$ and define $L_{m,\infty}:=\lim\limits_{n \to \infty}L_{m,n}$, $L_{\infty,n}:=\lim\limits_{m \to \infty}L_{m,n}$. Note that dominated converges and the Portmanteau-Theorem imply that $L=\lim\limits_{m \to \infty}L_{m,\infty}=\lim\limits_{n \to \infty}L_{\infty,n}$.
\\ \\
We need to show that for any $\epsilon > 0$ there exists $K \in \mathbb{N}$ such that $\forall m,n > K$: $\left| L - L_{m,n} \right| \leq \epsilon$. 
\\ \\
Because $\lim\limits_{n \to \infty}L_{\infty,n}=L$ we can find $K_1 \in \mathbb{N}$ such that $\forall n > K_1$: $\left| L - L_{\infty,n} \right| \leq \frac{\epsilon}{2}$.
Since $F$ is the distribution function of a discrete random variable there exists a finite set $A\subset E$, of cardinality $\left|A \right| < \infty$,  such that $\mu(A) > 1- \delta$ where $\delta$ is defined as $\delta:=\frac{\epsilon}{4(\left|A \right|+1)}$. Together, the 1\textsuperscript{st} auxiliary result \ref{par:conv1_aux1} and the fact that $A$ is finite, imply that there exists $K_2 \in \mathbb{N}$ such that $\forall m \geq K_2$: $\left| \mu(x)  - \mu_m(x) \right| < \delta$, $\forall x \in A$. Hence, we can deduce that $\forall m \geq K_2$: $\mu_m(A^c) \leq \mu(A^c)+\left|A \right|\delta\leq \delta+\left|A \right|\delta$.

Therefore, we can conclude that for all $m > K_2$, $n \in \mathbb{N}$:
\begin{align*}
&\left|  L_{\infty,n} - L_{m,n} \right| = \left| \int_{E} h_n(x)d\mu(x) - \int_{E} h_n(x)d\mu_m(x) \right| \\
&\leq \left|\int_{A} \underbrace{h_n(x) }_{\text{$\lVert h_n \rVert_{\infty} \leq 1$}} d (\mu(x) -\mu_m(x)) \right| + \left|\int_{A^c} \underbrace{h_n(x) }_{\text{$\lVert h_n \rVert_{\infty} \leq 1$}} d (\mu(x) -\mu_m(x)) \right| \\
& \leq \sum_{x \in A} \left| \mu(x)  - \mu_m(x) \right| + \mu(A^c) + \mu_m(A^c)\\
&\leq \left|A \right|\delta + \delta + \delta + \left|A \right|\delta = \delta(2\left|A \right| + 2 ) = \frac{\epsilon}{2} 
\end{align*}
Set $K = \max \left\lbrace K_1, K_2\right\rbrace$ and note that then $\forall m,n > K$:
\begin{align*}
\left| L - L_{m,n} \right| \leq \left| L - L_{\infty,n} \right| + \left|  L_{\infty,n} - L_{m,n} \right| \leq \frac{\epsilon}{2}+\frac{\epsilon}{2} = \epsilon
\end{align*}
\end{proof}
\end{proof}

We have seen that under the assumption of discrete marginals the MRA yields approximations of the unique tree dependent random vector. \\ \\
The MRA is a modification of the original reordering algorithm proposed by Arbenz et al. \cite{arbenz12}. We discussed the modifications in the beginning of Section \ref{ex:altReordering}. Due to the complexity and length of the above proof it might be difficult to identify the part of it where the modifications play a key role. We stress therefore that in particular the proof of the first step critically depends on them. The second step then exploits the fact that the marginals are discrete.
\\ \\
It currently remains open to show that the results also hold for more general marginals. In addition, it would be desirable to obtain similar results for the original reordering algorithm. The MRA has the disadvantage that its algorithmic efficiency decreases significantly with the number of levels in the aggregation tree. This is due to the second modification, which requires the algorithm to call itself repeatedly on each level to generate independent samples and, thus, leads to an exponential running time (in the number of levels). The original reordering algorithm, in contrast, does not suffer from this issue and will, therefore, most likely be favoured for practical applications.

\section{A word on the conditional independence \mbox{assumption}}
\label{sec:CIA}
The conditional independence assumption (\ref{def:CIA}) specifies the unique tree dependent distribution for a given aggregation tree model. One should definitely wonder, however, if this assumption is reasonable in practice. In this section we will briefly address this question.
\\ \\
Testing for conditional independence requires a large amount of data. Unfortunately, meaningful historical data is extremely rare in practice, and so we consider the following logical reasoning instead:
\\ \\
Assume the true risk distribution is given by $(X,Y,Z) \sim \mathcal{N}(\mu,\varSigma)$ with
\begin{align} \label{eq:CIA_true_risk}
\mu=\left( \begin{array}{c}
0 \\
0 \\
0
\end{array}\right) , \hspace{30 pt}
\varSigma = \left( \begin{array}{ccc}
1 & 0.5 & 0 \\
0.5 & 1 & 0 \\
0 & 0 & 1
\end{array}\right).
\end{align}
We first consider the aggregation tree model showed on the left of Figure \ref{fig:CIAexample} and denote by $(X^{td_1},Y^{td_1},Z^{td_1})$ its tree dependent distribution. It is easy to show that for this model the tree dependent distribution is equivalent to the true risk distribution, i.e. $(X^{td_1},Y^{td_1},Z^{td_1})  \sim \mathcal{N}(\mu,\varSigma)$.
\begin{figure}[h]
\centering
\includegraphics[width=1\linewidth]{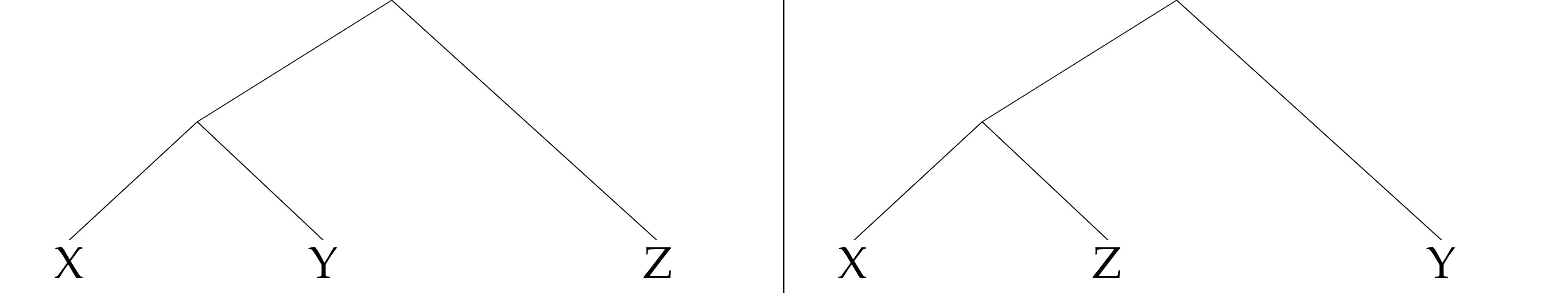}
\caption{Different ways of designing the tree structure.}
\label{fig:CIAexample}
\end{figure}

Consider now the aggregation tree model on the right of Figure \ref{fig:CIAexample}. Note that this model incorporates two bivariate normal copulas $C_{X,Z}$ and $C_{X+Z,Y}$ describing the dependence structure between $X$ and $Z$, respectively $X+Z$ and $Y$. By (\ref{eq:CIA_true_risk}), the copula $C_{X,Z}$ must be the independent copula, and the copula $C_{X+Z,Y}$ has correlation matrix
\begin{align*}
R_{X+Z,Y} = \left(  \begin{array}{cc}
1 & \frac{1}{2\sqrt{2}} \\
\frac{1}{2\sqrt{2}} & 1
\end{array}\right).
\end{align*}
Let $(X^{td_2},Y^{td_2},Z^{td_2})$ denote the tree dependent distribution associated with this model. A simple calculation yields that $(X^{td_2},Y^{td_2},Z^{td_2})\sim \mathcal{N}(\mu^{td_2},\varSigma^{td_2})$ with
\begin{align*}
\mu^{td_2}=\left( \begin{array}{c}
0 \\
0 \\
0
\end{array}\right) , \hspace{30 pt}
\varSigma^{td_2} = \left( \begin{array}{ccc}
1 & 0.25 & 0 \\
0.25 & 1 & 0.25 \\
0 & 0.25 & 1
\end{array}\right).
\end{align*}
Obviously, the distributions of $(X^{td_1},Y^{td_1},Z^{td_1})$ and $(X^{td_2},Y^{td_2},Z^{td_2})$ are not the same, and the example clearly shows that the tree dependent distribution is not invariant to a change of the tree structure.
\\ \\
In particular, this simple example illustrates that the conditional independence assumption should be treated with much caution: The tree dependent distribution critically depends on the tree structure. In practice, it will therefore hardly ever be an adequate approximation of the reality, unless the tree has exactly be designed in such a way that the data satisfies the conditional independence assumption. Easier said than done, considering that models with hundreds of individual risks are not unusual and meaningful data is often rare.

\chapter{The space of mildly tree dependent distributions}
\label{chap:spaceofmildly}
As we have mentioned already in the first chapters, an aggregation tree model (\ref{eq:triple}) does in general not uniquely specify the joint distribution of all risks. We called those distributions, which fit a given aggregation tree model, mildly tree dependent (recall Definition \ref{def:mildlytreedependent}).
\\ \\
Our lead Example \ref{ex:hierarchical} shall once more serve as motivation for this chapter. Assume that - similar as in the Concluding remark \ref{sec:concludingRemark} - an insurance company is also interested in the total risk they are taking in Switzerland, i.e. the risk $S=X_{1,1} + X_{2,1}$. As we know by now, the aggregation tree model does not fully specify this distribution. Since for instance the covariance Cov$(X_{1,1},X_{2,1})$ can vary depending on which mildly tree dependent distribution we look at, the variance $\text{Var}(S)=\text{Var}(X_{1,1}) +\text{Var}(X_{2,1}) + 2\text{Cov}(X_{1,1},X_{2,1})$  will vary as well. 
\begin{figure}[h]
\centering
\includegraphics[width=1\linewidth]{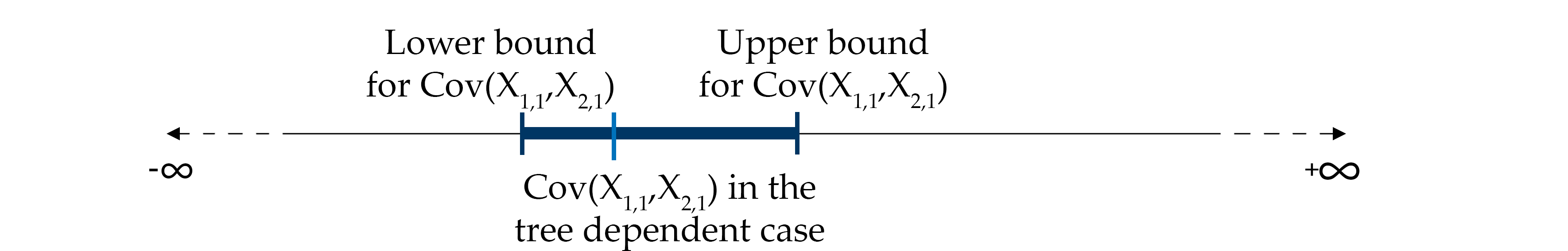}
\label{fig:toyexample}
\end{figure}

In the previous chapters we have seen that there is high evidence that the sample reordering algorithm yields approximations of the unique tree dependent distribution. However, as illustrated in Section \ref{sec:CIA}, it is questionable if the conditional independence assumption (\ref{def:CIA}), which specifies the tree dependent distribution,  is indeed a reasonable assumption in reality.
\\ \\
Therefore, it makes definitely sense to think about, for instance, how the range of values that Cov$(X_{1,1},X_{2,1})$ can take looks like. How large is it? Where does the tree dependent distribution lie in it? How do the different input parameters influence this range? 
\\ \\
In this chapter we will address these kind of questions. We will study simple trees for different sizes and shapes and we determine the space of mildly tree dependent distribution in various scenarios. These "toy-models" help us to understand the basic phenomena's at play while aggregating risks in this way. We believe that any insurer or reinsurer using such a tool should be aware of these systematic effects.

\section{The three-dimensional Gaussian tree}
In this first section we study extensively the three-dimensional Gaussian tree for which analytical results can be derived. The notation from the previous chapters is rather inconvenient for our purpose, and so we believe that a change of notation is justified. Consider the following model setup.

\subsubsection{Defining the model}
Our three-dimensional Gaussian tree incorporates three normally distributed random variables
\begin{align*}
X_1 \sim \mathcal{N}(0,\sigma_1^2), \hspace{20pt} X_2 \sim \mathcal{N}(0,\sigma_2^2), \hspace{20pt} X_3 \sim \mathcal{N}(0,\sigma_3^2),
\end{align*}
which represent the marginal risks. Note that w.l.o.g. we can assume that their means equal zero. Suppose further that we are given two bivariate normal copulas $C_{12}$ and $C_{\varnothing}$ with correlation matrices
\begin{align*}
R_{12} = \left(  \begin{array}{cc}
1 & \rho_{12} \\
\rho_{12} & 1
\end{array}\right) \qquad \text{and} \qquad 
R_{\varnothing} = \left(  \begin{array}{cc}
1 & \rho_{\varnothing} \\
\rho_{\varnothing} & 1
\end{array}\right).
\end{align*}
Here $C_{12}$ is the dependence structure that we impose between the risks $X_1$ and $X_2$, whereas $C_{\varnothing}$ is the dependence structure between the aggregated risk $X_1+X_2$ and $X_3$. The so defined aggregation tree model is illustrated in Figure \ref{fig:3dimgauss}.

\begin{figure}[h]
\centering
\includegraphics[width=0.8\linewidth]{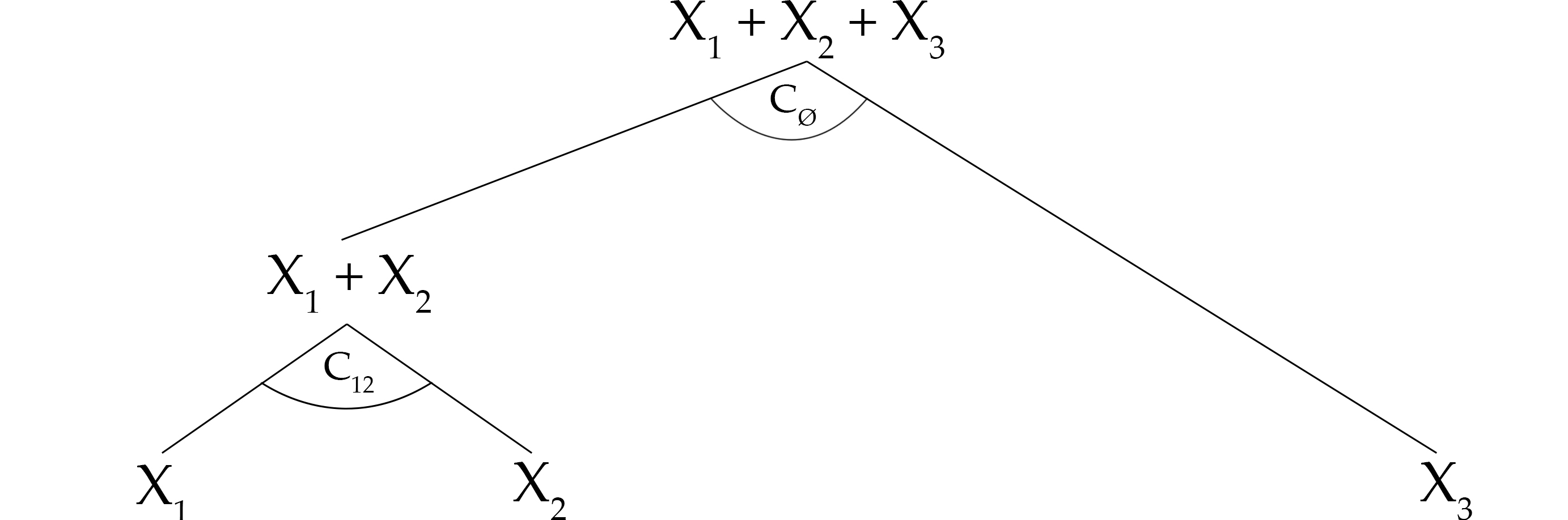}
\label{fig:3dimgauss}
\caption{Illustration of our aggregation tree model.}
\end{figure}

Note that with the above information the joint distribution of the random vector $(X_1,X_2,X_3)$ is not uniquely specified. If we further assume that the joint distribution is multivariate normal, it follows immediately that $(X_1,X_2,X_3) \sim \mathcal{N}(\mu,\varSigma)$ with
\begin{align*}
\mu=\left( \begin{array}{c}
0 \\
0 \\
0
\end{array}\right) , \hspace{30 pt}
\varSigma = \left( \begin{array}{ccc}
\sigma_1^2 & \rho_{12}\sigma_1\sigma_2 & \sigma_{13} \\
\rho_{12}\sigma_1\sigma_2 & \sigma_2^2 & \sigma_{23} \\
\sigma_{31} & \sigma_{32} & \sigma_3^2
\end{array}\right),
\end{align*}
where the covariances $\sigma_{13}=\sigma_{31}$ and $\sigma_{23}=\sigma_{32}$ enjoy a certain degree of freedom. We are, however, not completely free at choosing $\sigma_{13}$ and $\sigma_{23}$. A simple calculation yields that  $\sigma_{13} + \sigma_{23} = \rho_{\varnothing}\sqrt{\sigma_1^2+\sigma_2^2+2\rho_{12}\sigma_1\sigma_2}$ must hold in order to guarantee that the dependence structure between $X_1+X_2$ and $X_3$ is indeed given by $C_{\varnothing}$. A further constraint is that the resulting covariance matrix $\varSigma$ must be positive semidefinite.

\subsubsection{Determining the space of mildly tree dependent distributions}
Consider the above model where the parameters $\sigma_1$, $\sigma_2$, $\sigma_3$, $\rho_{12}$ and $\rho_{\varnothing}$ are given. Any multivariate normal distribution  $(X_1,X_2,X_3) \sim \mathcal{N}(\mu,\varSigma)$ which satisfies the above mentioned constraints
\begin{align}
\label{eq:constraint1}
&\sigma_{13} + \sigma_{23} = \rho_{\varnothing}\sqrt{\sigma_1^2+\sigma_2^2+2\rho_{12}\sigma_1\sigma_2}, \\
\label{eq:constraint2}
&\varSigma \qquad\text{pos. semidefinite},
\end{align}
is a mildly tree dependent distribution of our aggregation tree model. Our aim is to derive an analytic expression for the values the correlation coefficient $\rho_{13}:=\text{corr}(X_1,X_3)=\frac{\sigma_{13}}{\sigma_1\sigma_3}$ can take.
\\ \\
It turns out that the challenging part is to guarantee that the second constraint (\ref{eq:constraint2}) is satisfied. We exploit the fact that every positive semidefinite matrix $\varSigma$ has a Cholesky decomposition, i.e. there exists a lower triangular matrix $L$ such that $LL^{T}=\varSigma$:
\begin{align*}
&LL^T = \left( \begin{array}{ccc}
\ell_{11} & 0 & 0 \\
\ell_{21} & \ell_{22} & 0 \\
\ell_{31} & \ell_{32} & \ell_{33}
\end{array}\right)
\left( \begin{array}{ccc}
\ell_{11} & \ell_{21} & \ell_{31} \\
0 & \ell_{22} &\ell_{32} \\
0 & 0 & \ell_{33}
\end{array}\right) \\
&=
\left( \begin{array}{ccc}
\ell_{11}^2 & \ell_{11}\ell_{21} & \ell_{11}\ell_{31} \\
\ell_{11}\ell_{21} & \ell_{21}^2+\ell_{22}^2 & \ell_{21}\ell_{31}+\ell_{22}\ell_{32} \\
\ell_{11}\ell_{31} & \ell_{21}\ell_{31}+\ell_{22}\ell_{32} & \ell_{31}^2+\ell_{32}^2+\ell_{33}^2
\end{array}\right)
\stackrel{!}{=}\left( \begin{array}{ccc}
\sigma_1^2 & \rho_{12}\sigma_1\sigma_2 & \sigma_{13} \\
\rho_{12}\sigma_1\sigma_2 & \sigma_2^2 & \sigma_{23} \\
\sigma_{31} & \sigma_{32} & \sigma_3^2
\end{array}\right) = \varSigma.
\end{align*}
In this sense the problem translates into solving the following system of equations for the variables $\ell_{11}$, $\ell_{21}$, $\ell_{22}$, $\ell_{31}$, $\ell_{32}$, $\ell_{33}$:
\begin{align}
&\ell_{11}^2 = \sigma_1^2 \nonumber\\
&\ell_{11}\ell_{21} = \rho_{12}\sigma_1\sigma_2\nonumber\\
&\ell_{21}^2+\ell_{22}^2 = \sigma_2^2\nonumber\\
\label{eq:cha4_eq4}
&\ell_{31}^2+\ell_{32}^2+\ell_{33}^2 = \sigma_3^2\\
\label{eq:cha4_eq5}
&\ell_{11}\ell_{31} + \ell_{21}\ell_{31}+\ell_{22}\ell_{32} = \rho_{\varnothing}\sqrt{\sigma_1^2+\sigma_2^2+2\rho_{12}\sigma_1\sigma_2}
\end{align}
Note that the last equation (\ref{eq:cha4_eq5}) is the constraint (\ref{eq:constraint1}). The first three equations can be solved easily and we obtain $\ell_{11}=\pm \sigma_1$, $\ell_{21}=\pm \rho_{12}\sigma_2$ and $\ell_{22}=\pm \sqrt{\sigma_{2}^2(1-\rho_{12}^2)}$, where w.l.o.g. it is enough to consider the positive solutions only. \\
Next, we solve equation (\ref{eq:cha4_eq5}) for $\ell_{32}$:
\begin{align} \label{eq:chap4_ell32}
\ell_{32} = \frac{\rho_{\varnothing}\sqrt{\sigma_1^2+\sigma_2^2+2\rho_{12}\sigma_1\sigma_2}-\ell_{31}(\ell_{11}+\ell_{21})}{\ell_{22}} =: u - v \ell_{31},
\end{align}
where $u := \frac{\rho_{\varnothing}\sqrt{\sigma_1^2+\sigma_2^2+2\rho_{12}\sigma_1\sigma_2}}{\ell_{22}}$ and $v:= \frac{\ell_{11}+\ell_{21}}{\ell_{22}}$. We plug (\ref{eq:chap4_ell32}) into equation (\ref{eq:cha4_eq4}):
\begin{align*}
\ell_{31}^2+(u - v \ell_{31})^2+\ell_{33}^2 = \sigma_3^2,
\end{align*}
and after completing the square and a few simple algebraic operations we obtain the following equivalent equation:
\begin{align} \label{eq:chap4_ellipse}
\frac{(\ell_{31} - \frac{uv}{1+v^2})^2}{\frac{\sigma_3^2 - u^2}{1+v^2} + (\frac{uv}{1+v^2})^2} + \frac{\ell_{33}^2}{\sigma_3^2 - u^2 + \frac{(uv)^2}{1+v^2}} = 1.
\end{align}
Let $(x_0,y_0) := (\frac{uv}{1+v^2},0)$, $a:= \sqrt{\frac{\sigma_3^2 - u^2}{1+v^2} + (\frac{uv}{1+v^2})^2}$, $b:= \sqrt{\sigma_3^2 - u^2 + \frac{(uv)^2}{1+v^2}}$  and note that equation (\ref{eq:chap4_ellipse}) can then be written as
\begin{align}
\frac{(\ell_{31}-x_0)^2}{a^2} + \frac{(\ell_{33}-y_0)^2}{b^2} =  1,
\end{align}
which is the implicit equation of an ellipse with centre coordinates $(x_0,y_0)$ and axes $a$ and $b$.
\\ \\
We can conclude that $\ell_{31}$ can take any value in the interval $\left[ x_0 - a, x_0+a\right]$, and hence $\rho_{13}:=\frac{\sigma_{13}}{\sigma_1\sigma_3} = \frac{\ell_{11}\ell_{31}}{\sigma_1\sigma_3}$ can take any value in $\left[ \frac{\sigma_1(x_0 - a)}{\sigma_1\sigma_3}, \frac{\sigma_1(x_0 + a)}{\sigma_1\sigma_3}\right]$.
\\ \\
Comprehensive algebraic simplifications of the two endpoints $\rho_{13}^{min}:=\frac{\sigma_1(x_0 - a)}{\sigma_1\sigma_3}$ and $\rho_{13}^{max}:=\frac{\sigma_1(x_0 + a)}{\sigma_1\sigma_3}$ yield that in their simplest form they read as
\begin{align} \label{eq:chap4_minrho}
&\rho_{13}^{min} = \frac{\rho_{\varnothing}\rho_{12}\sigma_2 + \rho_{\varnothing}\sigma_1}{\sqrt{\sigma_1^2+\sigma_2^2+2\rho_{12}\sigma_1\sigma_2}} - \frac{\sqrt{\sigma_2^2(1-\rho_{12}^2)(1-\rho_{\varnothing}^2)}}{\sqrt{\sigma_1^2+\sigma_2^2+2\rho_{12}\sigma_1\sigma_2}}, \\
\label{eq:chap4_maxrho}
&\rho_{13}^{max} = \underbrace{\frac{\rho_{\varnothing}\rho_{12}\sigma_2 + \rho_{\varnothing}\sigma_1}{\sqrt{\sigma_1^2+\sigma_2^2+2\rho_{12}\sigma_1\sigma_2}}}_{=:\rho_{13}^{mid}} + \underbrace{\frac{\sqrt{\sigma_2^2(1-\rho_{12}^2)(1-\rho_{\varnothing}^2)}}{\sqrt{\sigma_1^2+\sigma_2^2+2\rho_{12}\sigma_1\sigma_2}}}_{=:\rho_{13}^{length}}.
\end{align}
In summary, we can say that for our three-dimensional Gaussian tree (Figure \ref{fig:3dimgauss}) the correlation coefficient $\rho_{13}:=\text{corr}(X_1,X_3)$ of a mildly tree dependent distribution can take any value in $\left[ \rho_{13}^{min},\rho_{13}^{max} \right]$. The midpoint of this interval is $\rho_{13}^{mid}$ and the length of the interval is given by $2\rho_{13}^{length}$.
\\ \\
Note that once we know the values $\rho_{13}$ (resp. $\sigma_{13}$) can take, it is easy to obtain analogous expressions for the values $\rho_{23}$ (resp. $\sigma_{23}$) can take, simply by exploiting the relation (\ref{eq:constraint1}). In this way we get a parametrization of all multivariate normal, mildly tree dependent distributions.
\\ \\
In the following examples we will further study and illustrate the expressions we have just derived for $\rho_{13}^{min}$ and $\rho_{13}^{max}$.

\subsubsection{Example 1}
In this first example we study the influence of the variances $\sigma_1^2$, $\sigma_2^2$ and $\sigma_3^2$ on the range $\left[ \rho_{13}^{min},\rho_{13}^{max} \right]$ of possible correlations between $X_1$ and $X_3$.

\paragraph{Case 1: $\boldsymbol{\sigma_1^2 \to \infty}$} We let the variance of $X_1$ go to infinity. It is to expect that since $X_1$ will dominate $X_2$, the influence of $X_2$ on the dependence structure between $X_1+X_2$ and $X_3$ will be negligible, i.e. $\text{corr}(X_1,X_3)\approx\text{corr}(X_1+X_2,X_3)=\rho_{\varnothing}$. It is easy to check that
\begin{align*}
\lim\limits_{\sigma_1^2 \to \infty}\rho_{13}^{min}=\lim\limits_{\sigma_1 \to \infty}\rho_{13}^{max}=\rho_{\varnothing}, 
\end{align*}
which confirms our intuition. Note that since the correlation coefficient associated with the unique tree dependent distribution (let us denote it by $\rho_{13}^{td}$) must lie somewhere inside the degenerated interval $\left[ \lim\limits_{\sigma_1^2 \to \infty}\rho_{13}^{min},\lim\limits_{\sigma_1^2 \to \infty}\rho_{13}^{max} \right] = \rho_{\varnothing}$, we can deduce that $\rho_{13}^{td}$ must be equal to $\rho_{\varnothing}$.

\paragraph{Case 2: $\boldsymbol{\sigma_2^2 \to \infty}$} We let the variance of $X_2$ go to infinity. In contrast to the previous case, the influence of $X_1$ on the dependence structure between $X_1+X_2$ and $X_3$ can be seen as negligible. Intuitively, we can expect that in this case the interval $\left[ \rho_{13}^{min},\rho_{13}^{max} \right]$ attains its maximum length. Our intuition is again confirmed by the analytical results:
\begin{align*}
&\lim\limits_{\sigma_2^2 \to \infty}\rho_{13}^{min}=\rho_{\varnothing}\rho_{12} - \sqrt{(1-\rho_{12}^2)(1-\rho_{\varnothing}^2)} \\
&\lim\limits_{\sigma_2^2 \to \infty}\rho_{13}^{max}=\rho_{\varnothing}\rho_{12} + \sqrt{(1-\rho_{12}^2)(1-\rho_{\varnothing}^2)}.
\end{align*}
Note that for fixed $\rho_{12}, \rho_{\varnothing}$ the maximum of $\rho_{13}^{length}$ is indeed $\sqrt{(1-\rho_{12}^2)(1-\rho_{\varnothing}^2)}$.

\paragraph{Case 3: $\boldsymbol{\sigma_3^2 \to \infty}$} We let the variance of $X_3$ go to infinity. From the expressions (\ref{eq:chap4_minrho}) \& (\ref{eq:chap4_maxrho}) we see immediately that $\rho_{13}^{min}$ and $\rho_{13}^{max}$ do not depend on $\sigma_3$.
\\ \\
To summarise, it can be noted that the range of values that $\rho_{13}^{max}$ can attain becomes larger if the variance of $X_1$ is small compared to the one of $X_2$. If the opposite is the case, then $\left[ \rho_{13}^{min},\rho_{13}^{max} \right]$ becomes smaller. 

\subsubsection{Example 2}
Let for the second example $\sigma_1\equiv\sigma_2\equiv\sigma_3\equiv1$ be fixed and equal to one. We plot the difference $\rho_{13}^{max}-\rho_{13}^{min}=2\rho_{13}^{length}$ between the upper and the lower bounds for $\rho_{13}$ in terms of $\rho_{12}$ and $\rho_{\varnothing}$.
\begin{figure}[h]
\centering
\includegraphics[width=1\linewidth]{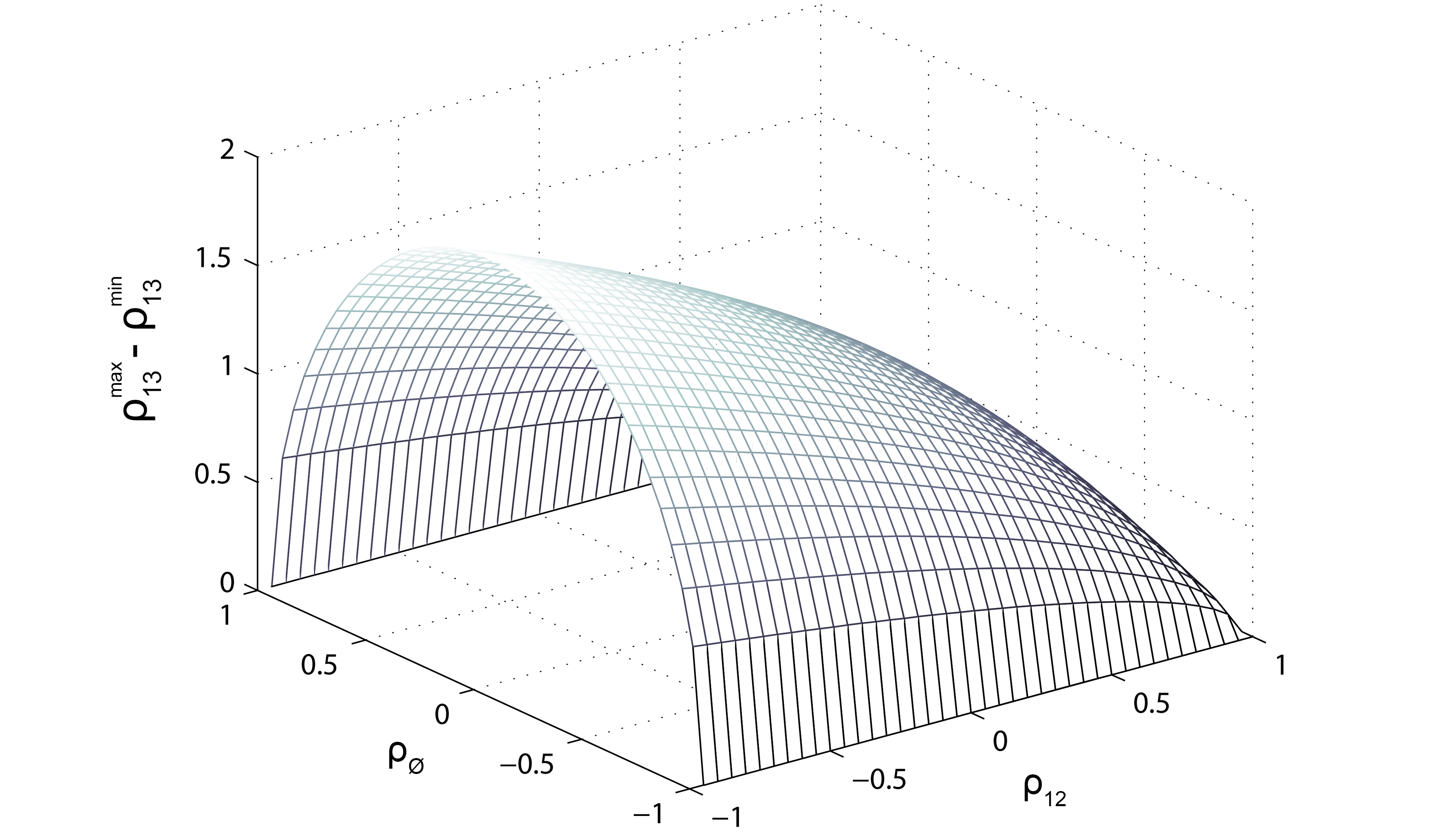}
\caption{Difference between the upper and the lower bounds for $\rho_{13}$ in Example 2.}
\label{fig:interval}
\end{figure}
\\ \\
The result is displayed in Figure \ref{fig:interval}. It seems that the space of mildly tree dependent distribution increases the more $X_1+X_2$ and $X_3$ become uncorrelated/independent ($\rho_{\varnothing} \to 0$). In the same time it also increases the more $X_1$ and $X_2$ become negatively correlated ($\rho_{12} \to -1$). The maximum is attained in $(\rho_{12},\rho_{\varnothing})=(-1,0)$ where it holds that $\rho_{13}^{max}-\rho_{13}^{min}=2$, i.e. $\rho_{13}$ can take any value in $\left[ \rho_{13}^{min},\rho_{13}^{max} \right]=\left[-1,1 \right] $.  We will further study this effect in the next example.

\subsubsection{Example 3}
We study the same scenario as in the previous example ($\sigma_1\equiv\sigma_2\equiv\sigma_3\equiv1$). This time, we plot the interval $\left[ \rho_{13}^{min},\rho_{13}^{max} \right]$ against $\rho_{\varnothing}$ for nine different values of $\rho_{12}$. The resulting plots are displayed in Figure \ref{fig:ninetimesnine}. Note that we choose to display also the values of the tree dependent correlation coefficient $\rho_{13}^{td}$.
\begin{figure}[h]
\centering
\includegraphics[width=0.9\linewidth]{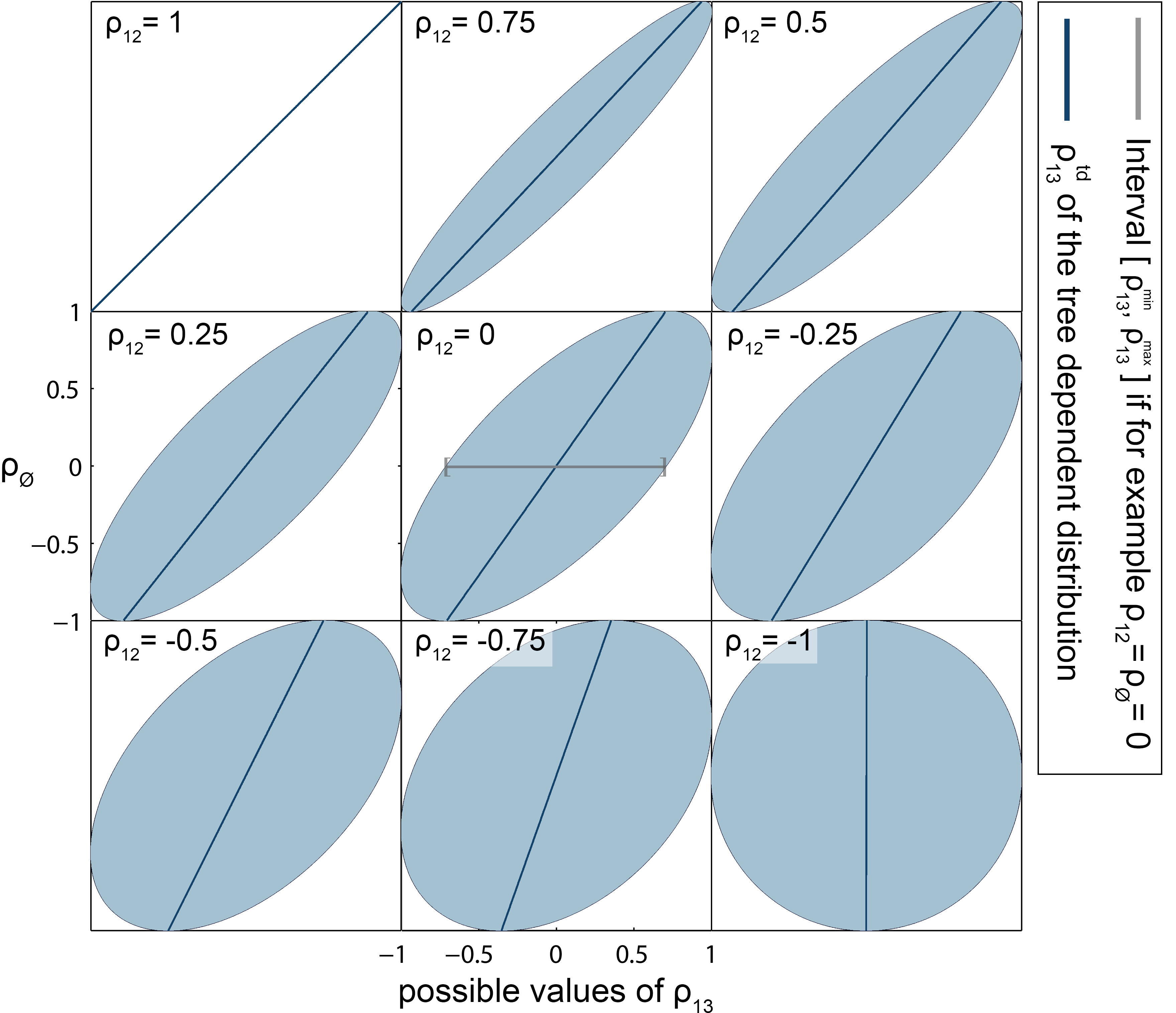}
\caption{Illustration of the interval $\left[ \rho_{13}^{min},\rho_{13}^{max} \right]$ in terms of $\rho_{12}$ and $\rho_{\varnothing}$ in Example 3.}
\label{fig:ninetimesnine}
\end{figure}
\\ \\
It is not surprising that we can observe the same effect as in the previous example. In addition, the illustration also reveals how position and length of the interval $\left[ \rho_{13}^{min},\rho_{13}^{max} \right]$ change in terms of $\rho_{12}$ and $\rho_{\varnothing}$.
\\ \\
In case of total correlation between $X_1$ and $X_2$ ($\rho_{12}=1$), the interval $\left[ \rho_{13}^{min},\rho_{13}^{max} \right]$ becomes degenerated and coincides with $\rho_{13}^{td}$ of the unique tree dependent distribution. In fact, $\left[ \rho_{13}^{min},\rho_{13}^{max} \right]$ seems to be symmetric around $\rho_{13}^{td}$, i.e. $\rho_{13}^{td}$ is equal to the midpoint $\rho_{13}^{mid}$ of the interval. A simple calculation confirms our assumption.
\\
We stress that this effect can not be observed in general. For a four-dimensional Gaussian tree the equality $\rho_{13}^{td}=\rho_{13}^{mid}$ would, for instance, not hold anymore, and so $\rho_{13}^{td}$ is closer either to the left or to the right endpoint of the interval $\left[ \rho_{13}^{min},\rho_{13}^{max} \right]$.

\subsubsection{Example 4}
The previous examples have shown that the space of mildly tree dependent distributions may be particularly large under some circumstances. In practice, we would probably prefer a model where the space of mildly tree dependent distributions is small. Because then the mildly tree dependent distributions are concentrated around the true joint risk distribution, and so each of them yields already a good approximation of said true risk distribution.
\\ \\
In this example we show how the choice of an appropriate aggregation tree model can considerably narrow the space of mildly tree dependent distributions. Let for this purpose the true risk distribution be given by $(X_1,X_2,X_3) \sim \mathcal{N}(\mu,\varSigma)$ with
\begin{align*}
\mu=\left( \begin{array}{c}
0 \\
0 \\
0
\end{array}\right) , \hspace{30 pt}
\varSigma = \left( \begin{array}{ccc}
1 & -1 & 1 \\
-1& 2 & -1 \\
1 & -1 & 1
\end{array}\right).
\end{align*}
Assume two insurers (Insurer 1 \& Insurer 2) wish to fit a three-dimensional Gaussian aggregation tree model. Their approach is, however, not exactly the same, as illustrated in Figure \ref{fig:insurers}.
\begin{figure}[h]
\centering
\includegraphics[width=1\linewidth]{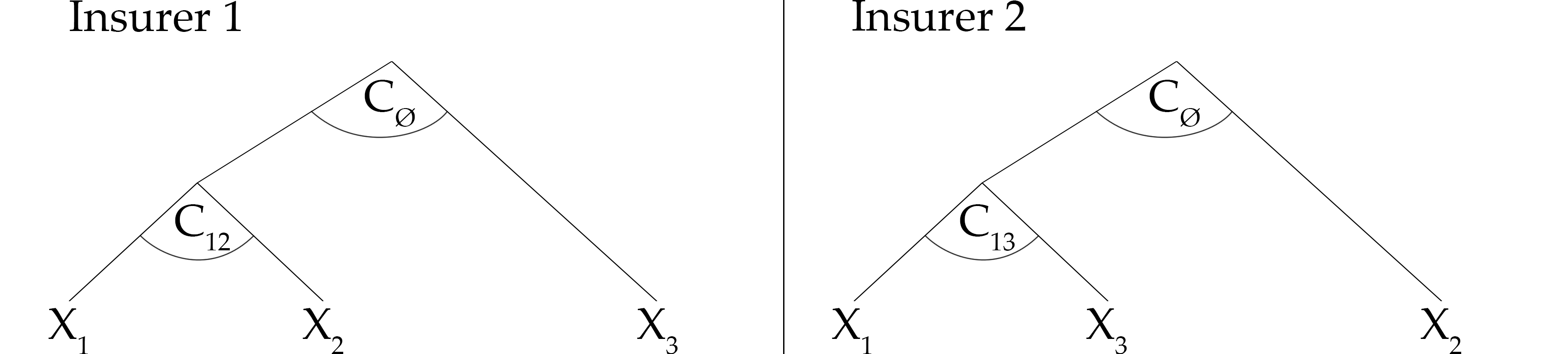}
\caption{Illustration of the aggregation tree models of Insurer 1 and Insurer 2.}
\label{fig:insurers}
\end{figure}

\paragraph{Insurer 1} Insurer 1 decides to fit an aggregation tree as illustrated on the left of Figure \ref{fig:insurers}. Hence he imposes a normal copula $C_{12}$ between $X_1$ and $X_2$ and another normal copula $C_{\varnothing}$ between $X_1+X_2$ and $X_3$. \\
From the information available (data \& expert opinion) he obtains the following exact estimation of the marginal risks: $X_1 \sim \mathcal{N}(0,1)$, $X_2 \sim \mathcal{N}(0,2)$ and $X_3 \sim \mathcal{N}(0,1)$. He further estimates that the normal copulas $C_{12}$ and $C_{\varnothing}$ have correlation matrices
\begin{align*}
R_{12} = \left(  \begin{array}{cc}
1 & -\frac{1}{\sqrt{2}} \\
-\frac{1}{\sqrt{2}} & 1
\end{array}\right) \qquad \text{and} \qquad 
R_{\varnothing} = \left(  \begin{array}{cc}
1 & 0 \\
0 & 1
\end{array}\right).
\end{align*}
With the formulas (\ref{eq:chap4_minrho}) \& (\ref{eq:chap4_maxrho}) one can easily check that for this aggregation tree model the space of normal, mildly tree dependent distributions becomes maximal. A parametrization of this space would for instance be given by $\boldsymbol{X}_a \sim \mathcal{N}(\mu_a,\varSigma_a)$ with
\begin{align*}
\mu_a\equiv\left( \begin{array}{c}
0 \\
0 \\
0
\end{array}\right) , \hspace{30 pt}
\varSigma_a = \left( \begin{array}{ccc}
1 & -1 & a \\
-1& 2 & -a \\
a & -a & 1
\end{array}\right), \quad a \in \left[ -1,1\right].
\end{align*}

\paragraph{Insurer 2} Insurer 2 decides to fit an aggregation tree as illustrated on the right of Figure \ref{fig:insurers}. Hence, he imposes a normal copula $C_{13}$ between $X_1$ and $X_3$ and another normal copula $C_{\varnothing}$ between $X_1+X_3$ and $X_2$. \\
From the information available (data \& expert opinion), he obtains the following exact estimation of the marginal risks: $X_1 \sim \mathcal{N}(0,1)$, $X_2 \sim \mathcal{N}(0,2)$ and $X_3 \sim \mathcal{N}(0,1)$. He further estimates that the normal copulas $C_{13}$ and $C_{\varnothing}$ have correlation matrices
\begin{align*}
R_{13} = \left(  \begin{array}{cc}
1 & 1 \\
1 & 1
\end{array}\right) \qquad \text{and} \qquad 
R_{\varnothing} = \left(  \begin{array}{cc}
1 & -\frac{1}{\sqrt{2}} \\
-\frac{1}{\sqrt{2}} & 1
\end{array}\right).
\end{align*}
In this case, the aggregation tree model fully specifies the joint distribution and the space of mildly tree dependent distributions is minimal. The only distribution contained in it is the true risk distribution $(X_1,X_2,X_3) \sim \mathcal{N}(\mu,\varSigma)$.
\\ \\
Note that both approaches are correct and would yield the same distribution for the aggregated risk $X_1+X_2+X_3$. The difference is that, although both insures need to estimate the same number of parameters, the model of "Insurer 2" fully specifies the joint distribution, whereas the model of "Insurer 1" does not.

In summary, we have shown in detail how the space of mildly tree dependent distribution is affected by the different input parameters of the aggregation tree model. The results of this section suggest that the aggregation tree should, whenever possible, be constructed in such a way that the risks in the sub-portfolios do have a strong positive correlation. By doing so, the space of mildly tree dependent can be considerably narrowed, as Example 4 has impressively illustrated.
\\ \\
The idea of grouping together highly correlated risks in sub-portfolios when setting up the aggregation tree is not new and has already been proposed by Côté \& Genest \cite{genest15} and Bruneton \cite{bruneton11}. Their arguments are, however, rather of qualitative and practical nature. To our knowledge, the insights gained in this section are for the first time providing also a profound mathematical justification for this approach.

\section{Generalization to non-Gaussian trees}
In this short section we discuss the space of mildly tree dependent distributions in more general settings. Subject of the first subsection is again the three-dimensional tree. In contrast to the previous section we do, however, not limit ourselves to Gaussian trees only. In the second subsection we study some properties of an eight-dimensional, symmetric aggregation tree.

\subsection{The general three-dimensional tree}
Consider the same aggregation tree structure as in Figure \ref{fig:3dimgauss}. Regarding the distributions of $X_1$, $X_2$ and $X_3$ the only assumption is that $\sigma_i^2=\text{Var}(X_i)< \infty$, for $i \in \left\lbrace 1,2,3\right\rbrace$. We further assume that the copulas $C_{12}$ and $C_{\varnothing}$ are such that $\text{corr}(X_1,X_2)=\rho_{12}$ and $\text{corr}(X_1+X_2,X_3)=\rho_{\varnothing}$.
\\ \\
Suppose we want to determine the set of covariance matrices 
\begin{align*}
\varSigma = \left( \begin{array}{ccc}
\sigma_1^2 & \rho_{12}\sigma_1\sigma_2 & \sigma_{13} \\
\rho_{12}\sigma_1\sigma_2 & \sigma_2^2 & \sigma_{23} \\
\sigma_{31} & \sigma_{32} & \sigma_3^2
\end{array}\right),
\end{align*}
which are compliant with our model (in the sense that there exists a mildly tree dependent distribution with covariance matrix $\varSigma$). Obviously, this is exactly the same problem as in the previous section, and hence it all comes down to determining those $\varSigma$ for which the constraints (\ref{eq:constraint1}) and (\ref{eq:constraint2}) are satisfied.
\\ \\
We can proceed as in the previous section to obtain the same expressions (\ref{eq:chap4_minrho}) and (\ref{eq:chap4_maxrho}) for the lower and upper bound for $\rho_{13}:=\text{corr}(X_1,X_3)$. Once we have this, it is straightforward to obtain a parametrization for the $\varSigma$'s.
\\ \\
We would like to emphasize that in the previous section, in fact, we never used the Normality of the tree. In particular, all the results and examples from the previous section are still true and remain the same, irrespective of whether the tree is Gaussian or not.

\subsection{An eight-dimensional, symmetric tree}
In this subsection we consider a perfectly symmetric tree involving eight identically distributed risks with mean 0 and variance 1, as illustrated in Figure \ref{fig:eightdim}. Note that this model incorporates seven bivariate copulas. We assume that each of these seven copulas is such that the correlation between the two risks linked through the respective copula is always equal to $\rho \in \left[  -1,1\right] $, for instance $\text{corr}(X_3,X_4)=\rho$, as well as $\text{corr}(S_5,S_6)=\rho$.
\begin{figure}[h]
\centering
\includegraphics[width=1\linewidth]{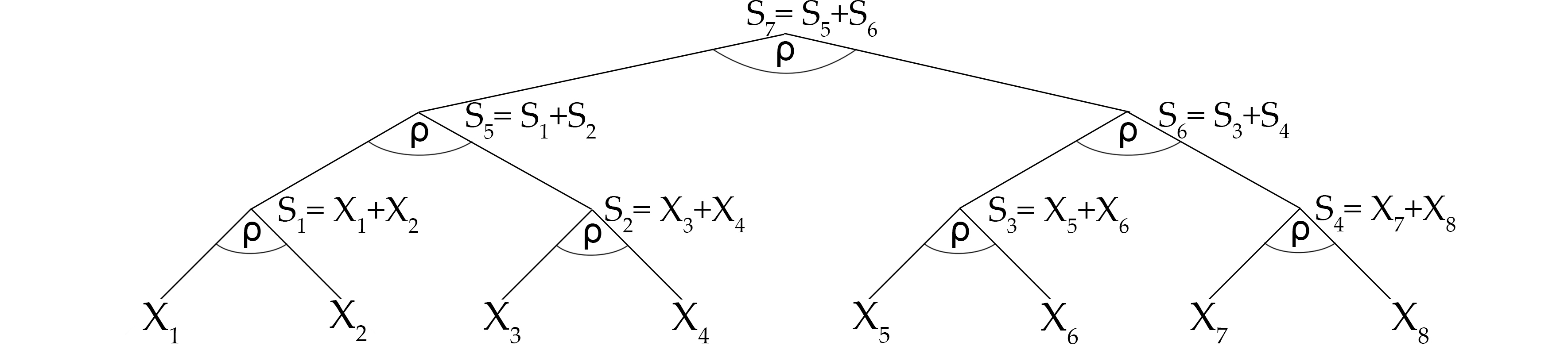}
\caption{Illustration of the eight-dimensional, symmetric tree.}
\label{fig:eightdim}
\end{figure}
\subsubsection{The tree dependent distribution}
This special setup allows us to derive an expression for the correlation between two individual risks in case the conditional independence assumption (\ref{def:CIA}) is satisfied (also see \cite{genest15}, Example 2).
\\ \\
Assume (\ref{def:CIA}) holds and $(X_1, \ldots,X_8)$ is the unique tree dependent distribution. We want to calculate $\rho_{13}^{td}:=\text{corr}(X_1,X_3)$. By (\ref{def:CIA}) and the tower property,
\begin{align*}
\text{corr}(X_1,X_3)&= E[X_1X_3]=E[E[X_1X_3\mid S_1,S_2]] \\
&=E[E[X_1\mid S_1,S_2]E[X_3\mid S_1,S_2]] \\
&=E[E[X_1\mid S_1]E[X_3\mid S_2]].
\end{align*}
Since $E[X_1\mid S_1]+E[X_2\mid S_1]=E[X_1+X_2\mid S_1]=X_1+X_2=S_1$ and $E[X_1\mid S_1]=E[X_2\mid S_1]$ by symmetry, we can deduce that $E[X_1\mid S_1]=\frac{S_1}{2}$. Analogously, it holds that $E[X_3\mid S_2]=\frac{S_2}{2}$. Hence,
\begin{align*}
\text{corr}(X_1,X_3)&= E[\frac{S_1}{2}\frac{S_2}{2}]=\frac{1}{4}E[S_1S_2]=\frac{1}{4}\rho\text{Var}(S_1)\text{Var}(S_2)=\frac{\rho(1+\rho)}{2}.
\end{align*}
Analogously, we get that $\rho_{18}^{td}:=\text{corr}(X_1,X_8)=\frac{\rho(1+\rho)^2}{4}$. In fact, for a perfectly symmetric tree involving $2^{\ell}$ identically distributed risks a simple induction argument yields that the correlation between parents of degree 
$k \in \left\lbrace 1,\ldots,2^{\ell} \right\rbrace$ is $\frac{\rho(1+\rho)^{k-1}}{2^{k-1}}$. For $\rho \in \left( -1,1\right)$ this directly implies that in the limit, the correlation between parents of degree $k$ tends to 0 as $k \to \infty$.

\subsubsection{The mildly tree dependent distributions}
Let now $(X_1, \ldots,X_8)$ be a mildly tree dependent distribution of our aggregation tree model. Our goal is to determine the lower and upper bounds, $\rho_{13}^{min}$ and $\rho_{13}^{max}$, for the correlation $\text{corr}(X_1,X_3)$, as well as the bounds, $\rho_{18}^{min}$ and $\rho_{18}^{max}$, for $\text{corr}(X_1,X_8)$. 
\\ \\
The problem reduces to determine the set of covariance matrices satisfying a couple of constraints similar to those in (\ref{eq:constraint1}) and (\ref{eq:constraint2}).
Again, the main difficulty is to guarantee that the covariance matrix of $(X_1, \ldots,X_8)$ is positive semidefinite. Unlike the three-dimensional case, analytical results can not be derived so easily.
\\ \\
Fortunately, the problem can be easily formulated as a semidefinite programming problem for which sufficiently exact approximations exist. For a good introduction to semidefinite programming see \cite{freund04}. We use the free MATLAB software SDPT3 from \cite{Toh99} to solve the semidefinite programming problem.
\begin{figure}[h]
\centering
\includegraphics[width=1\linewidth]{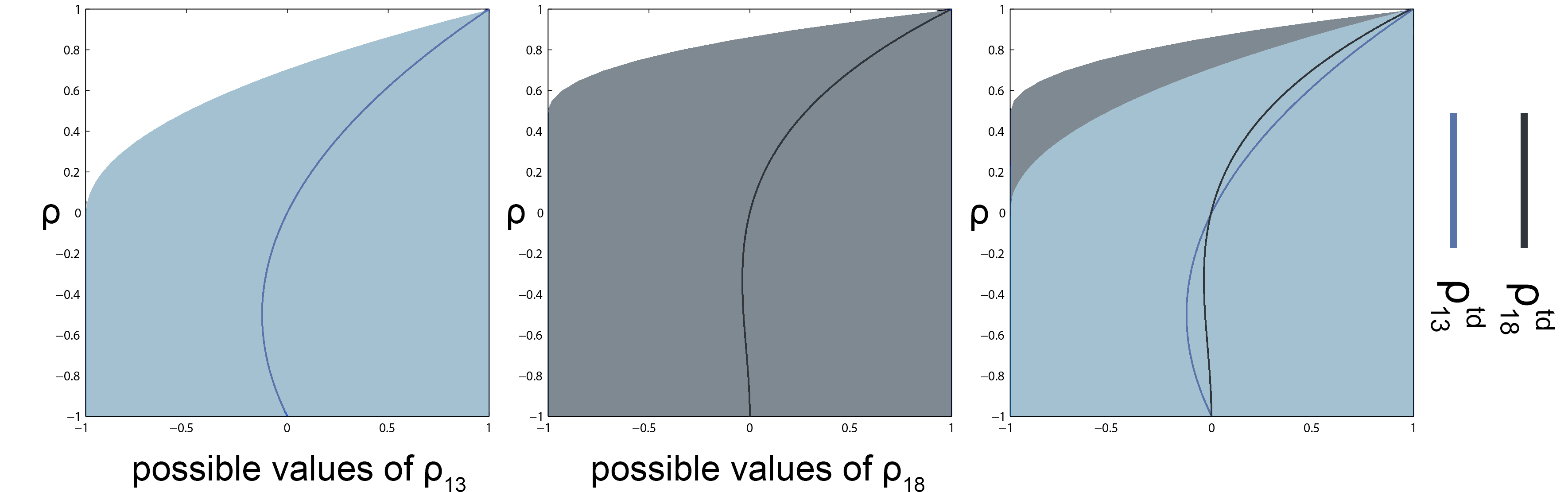}
\caption{Illustration of the results for the eight-dimensional, symmetric tree.}
\label{fig:8dim}
\end{figure}

In Figure \ref{fig:8dim} we display the results. The light blue plot on the left shows the interval $\left[ \rho_{13}^{min},\rho_{13}^{max}\right]$, as well as the value of $\rho_{13}^{td}$, in terms of $\rho \in \left[ -1,1\right]$. For $\rho < 0$, the correlation $\rho_{13}$ of a mildly tree dependent distributions could take any value in $[-1,1]$. For $\rho \geq 0$, the length of $\left[ \rho_{13}^{min},\rho_{13}^{max}\right]$ is decreasing in $\rho$. 
\\
A similar picture is found in the middle plot where we display the interval $\left[ \rho_{18}^{min},\rho_{18}^{max}\right]$, as well as the value of $\rho_{18}^{td}$, in terms of $\rho \in \left[ -1,1\right]$.
\\ \\
In the right plot we finally display the two previous plots together. It is remarkable that the left plot is completely included in the middle plot, i.e. for a fixed value of $\rho$, $\text{corr}(X_1,X_8)$ does enjoys more freedom than $\text{corr}(X_1,X_3)$. The result is not surprising: It should be intuitively clear that the dependence structure becomes less and less specified the further away two risks are laying from each other in the aggregation tree. Figure \ref{fig:8dim} just confirms our assumption.

\chapter{Conclusion}
\label{cha:6}
The present thesis provides a self-contained discussion of copula-based hierarchical risk aggregation. A strong focus is placed on the joint distributions associated with an aggregation tree model.
\\ \\
Chapter \ref{cha:2} and \ref{cha:3}  introduced the aggregation tree model as well as the sample reordering algorithm developed by Arbenz et al. \cite{arbenz12}. The aggregation tree model, consisting of a tree structure, copulas, and marginal distributions, offers a very flexible way to model risks. The method does not require to explicitly specify the dependence structure between all risks, but requires to do so only for different (sub-)aggregates and risk classes. As a consequence, it does not lead to a unique joint distribution of all risk. This fact is not critical as long as the ultimate goal of the user is to develop a model for the total aggregate. 
\\
In practice, however, some applications require the joint distribution of the individual risk. Against this background, we focused on two main questions.
\\ \\
The subject of the first question was the sample reordering algorithm. The original algorithm proposed by Arbenz et al. yields samples that approximate the total aggregate; yet hardly anything is known about the \textit{joint distribution} approximated by the samples. We have shown that there is numerical evidence that the samples may approximate the unique tree dependent distribution. This has encouraged us to develop a modified reordering algorithm (MRA). Under the additional constraint of discrete marginals the MRA yields an approximation of the unique tree dependent distribution. The proof mainly exploits the fact that the reordered samples satisfy the conditional independence assumption, which specifies the unique tree dependent distribution. It must be said, though, that the practical relevance of the MRA is questionable due to its poor algorithmic efficiency when the tree structure contains many aggregation levels. The question whether similar results can be proved for the original algorithm remains an interesting topic for future research.
\\ \\
It is currently not sure whether the conditional independence assumption is indeed a reasonable assumption in practice. The logical consideration of this question in Section \ref{sec:CIA} suggest that the assumption should at least be treated with much caution. In this respect, it is important to have a better understanding of the space of mildly tree dependent distributions. This was the subject of the second question. On the simple example of a three-dimensional Gaussian aggregation tree model we illustrated how the different parameters of the aggregation tree model affect the space of mildly tree dependent distributions. In particular, we showed that the space can be considerably narrowed by choosing a favourable aggregation order. The systematic effects observed in this and more general examples finally also allowed conclusions of the way the aggregation tree should be designed.

\appendix

\chapter{Appendix}
\label{appendix}
In Section \ref{sec:convergence2} we proved that when the MRA is applied to a 3-dimensional aggregation tree with discrete marginals, it yields samples that approximate the unique tree dependent distribution. \\
In this appendix we briefly discuss how the result can be generalized for an arbitrary aggregation tree. We stress that what follows is indeed a discussion rather than a complete and rigorous proof.
\\ \\
Let $\left( \tau,(F_I)_{I\in \mathscr{L}(\tau)},(C)_{I \in \mathscr{B}(\tau)}\right)$ be an arbitrary aggregation tree model with discrete marginals. We denote by $\boldsymbol{X}_{\varnothing}:=(X_{J})_{J \in \mathscr{LD}(\varnothing)}$ its unique tree dependent distribution and for $I \in \mathscr{B}$, $\boldsymbol{X}_{I}:=(X_J)_{J \in \mathscr{LD}(I)}$ denotes the restriction of $\boldsymbol{X}_{\varnothing}$ on the components $\mathscr{LD}(I)$.  We want to show that the MRA-reordered random vectors converge in distribution towards the unique tree dependent distribution associated the model. The idea is to use induction from the bottom to the top of the tree.
\\ \\
Fix hence $I \in \mathscr{B}$ and denote by $\boldsymbol{X}_{I}^{(1:n)},\ldots,\boldsymbol{X}_{I}^{(n:n)}$ the reordered random vectors (\ref{eq:algArbenz2}) obtained in step 3 of the MRA for the branching node $I$. The induction hypothesis implies that $\boldsymbol{X}_{I,i}^{(k:n)} \xrightarrow{n \to \infty} \boldsymbol{X}_{I,i}$ in distribution, for $k \in \left\lbrace 1,\ldots,n \right\rbrace$ and $i = 1,\ldots,N_I$.
\\ \\
Recall that we divided the proof of Theorem \ref{the:conv2} into two steps. The first step was to show that the $\boldsymbol{X}_{I}^{(1:n)},\ldots,\boldsymbol{X}_{I}^{(n:n)}$ satisfy the conditional independence assumption. We will assume for the moment that this is true. The second step can be easily generalized for an arbitrary aggregation tree model. The crucial point is to note that
\begin{align*}
P[X_J \leq x_J : J \in \mathscr{LD}(I)] &= E\left[ E[\mathbb{1}_{\left\lbrace X_J \leq x_J\right\rbrace }: J \in \mathscr{LD}(I)\mid X_{I,1},\ldots,X_{I,N_I}]\right]  \\
&= E\left[  \prod_{k=1}^{N_I} E[\mathbb{1}_{\left\lbrace X_J \leq x_J\right\rbrace }: J \in \mathscr{LD}(I,k)\mid X_{I,k}]\right], 
\end{align*}
due to the conditional independence assumption and the tower-property. Then we can proceed as in the second step and use the induction hypothesis to obtain the desired result.
\\ \\
The extremely tedious part is to show that $\boldsymbol{X}_{I}^{(1:n)},\ldots,\boldsymbol{X}_{I}^{(n:n)}$ indeed satisfy the conditional independence assumption. Recall that this was already notationally involved for the simple 3-dimensional tree. Although we are more than confident that the result can be generalized, it was not the subject of this thesis to provide a rigorous proof thereof.

\backmatter
\addtocontents{toc}{\protect\enlargethispage{\baselineskip}}
\bibliographystyle{plain}
\bibliography{refs}


\end{document}